\documentclass[10pt,journal]{IEEEtran}
\usepackage{amsmath,amsfonts}
\usepackage{algorithm}
\usepackage{array}
\usepackage[caption=false,font=normalsize,labelfont=sf,textfont=sf]{subfig}
\usepackage{textcomp}
\usepackage{stfloats}
\usepackage{url}
\usepackage{verbatim}
\usepackage{graphicx}
\usepackage{cite}

\usepackage[utf8]{inputenc}
\usepackage[english]{babel}
\usepackage{amsfonts}
\usepackage{qcircuit}

\usepackage[T1]{fontenc}
\usepackage{xkeyval}
\usepackage{tikz}
\usepackage{placeins}
\usepackage{graphicx}
\usepackage[export]{adjustbox}
\usepackage{soul}
\usepackage{dcolumn}
\usepackage{bm}
\usepackage{tikz}
\usetikzlibrary{calc}
\usepackage{caption}
\usepackage{subcaption}
\usepackage[margin=1in]{geometry}
\usepackage{amssymb}
\usepackage{overpic}
\usepackage{epstopdf}
 \usepackage[usenames,dvipsnames]{pstricks}
 \usepackage{epsfig}
 \usepackage{pst-grad} 
 \usepackage{pst-plot} 
 \usepackage[space]{grffile} 
 \usepackage{etoolbox} 
 \makeatletter 
 \patchcmd\Gread@eps{\@inputcheck#1 }{\@inputcheck"#1"\relax}{}{}
 \makeatother
\usepackage{booktabs}
\usepackage{boxhandler}
\usepackage{amsmath}
\usepackage{braket}
\usepackage{amsmath}
\usepackage{qcircuit}
\usepackage{float}
\makeatletter
\let\newfloat\newfloat@ltx
\makeatother
\usepackage{algpseudocode}
\usepackage{tikz}

\usepackage{hyperref}
\hypersetup{
  colorlinks   = true,    
  urlcolor     = blue,    
  linkcolor    = blue,    
  citecolor    = blue      
}

\usepackage{multirow}
\usepackage{enumitem}
\usepackage[normalem]{ulem}

\usepackage[amsthm]{ntheorem}

\newcommand{\F}{\mathbb{F}}

\theoremstyle{plain}
\newtheorem{theorem}{Theorem}[section]
\newtheorem{lemma}[theorem]{Lemma}
\newtheorem{corollary}[theorem]{Corollary}
\newtheorem{proposition}[theorem]{Proposition}

\theoremstyle{definition}

\newtheorem{remark}[theorem]{Remark}

\numberwithin{equation}{section}
\newtheorem{example}[theorem]{Example}

\newcommand{\Span}{{\rm Span }}
\newcommand{\wt}{{\rm wt}}

\begin{document}

\title{Generalized Bicycle Codes with Low Connectivity: Minimum Distance Bounds and Competitive Circuit-Level Performance}

\author{Reza Dastbasteh, Olatz {Sanz Larrarte}, Arun John Moncy, Paul Schnabl, Pedro M. Crespo, Josu {Etxezarreta Martinez}, and Ruben M. Otxoa%
\thanks{R. Dastbasteh, O. Sanz Larrarte, A. J. Moncy, P. M. Crespo, and J. Etxezarreta Martinez are with the Department of Basic Sciences, Tecnun--University of Navarra, San Sebastian, Spain }%
\thanks{A. J. Moncy is also with the Donostia International Physics Center, San Sebastian, Spain.}%
\thanks{P. Schnabl is with the Institute for Theoretical Physics, University of Innsbruck, Austria.}%
\thanks{R. M. Otxoa is with the Hitachi Cambridge Laboratory, J. J. Thomson Avenue, Cambridge, United Kingdom.}%
}

\maketitle

\begin{abstract}
We present new upper and lower bounds on the minimum distance of certain generalized bicycle (GB) codes that go beyond the capabilities of techniques based on classical codes. These bounds are shown to capture the true minimum distance for specific subclasses within the GB family.
In particular, we demonstrate the existence of two highly degenerate GB code families with parameters $[\![d^2+1,2,d]\!]_2$ for odd $d \geq 3$ and $[\![d^2,2,d]\!]_2$ for even $d \geq 4$, both having qubit connectivity of four, similar to that of surface codes.
For the odd-distance family, we analyze the structure of low-weight logical Pauli operators and establish the existence of a fault-tolerant logical CNOT gate between the two logical qubits, achievable via a simple relabeling of data qubits.
Next, using nearly self-inverse polynomials, we extend our results to construct an infinite class of connectivity-six GB codes with an improved rate. For GB codes based on these polynomials, we further present a technique to realize transversal logical CNOT gates via data-qubit permutations.
We further construct syndrome extraction schemes for our connectivity four GB families that avoid reductions in minimum distance due to circuit-level faults propagating from check qubits to data qubits.
Finally, we numerically evaluate the logical error rate performance of the mentioned connectivity four families of GB codes under circuit-level noise, obtaining thresholds comparable to those of rotated surface codes.
\end{abstract}


\section{Introduction}

Quantum error correction (QEC) is a crucial technique for unlocking the theoretical potential of quantum computing, given the extreme vulnerability of quantum information. 
Among all quantum error-correcting codes, the family of quantum low-density parity-check (LDPC) codes with low qubit-check connectivity enables low-overhead error correction, making them a promising candidate for large-scale, fault-tolerant quantum computation. 
There have been several proposals for quantum LDPC codes that aim to balance theoretical guarantees with hardware feasibility
\cite{breuckmann2016constructions, higgott2308constructions,kovalev2013quantum, panteleev2021degenerate, lin2024,tillich2013quantum, bravyi2024high,kovalev2011design,kovalev2012}.
In particular, recently it was shown that certain quantum LDPC codes can offer a constant encoding rate and linear distance scaling
\cite{DinurAsymptoticallyGood,gooLDPChomological,LeverrierAsymptoticallyGood,hsiehGood,physicsGood1,physicsGood2}. 

One of the main challenges associated with general quantum LDPC codes is the difficulty of constructing codes that have very low qubit interaction, short lengths, and good parameters \cite{wang2022}. 
In particular, the block size of theoretically good codes is often large, requiring many qubits for implementation. 

Therefore, shorter codes with a small, fixed number of qubit interactions (fewer than ten for most physical platforms) and greater hardware compatibility, even if their parameters are suboptimal, can be more practical in certain applications.

GB codes form a family of quantum codes that are particularly well-suited for constructing short codes with several desirable properties \cite{kovalev2013quantum,panteleev2021degenerate}. 
These include ease of construction, connections to other code families (such as quantum hypergraph-product and bivariate bicycle codes) \cite{tillich2013quantum,bravyi2024high}, linear distance scaling, efficient decoding, redundancy of minimum-weight stabilizer generators, and various other advantageous features \cite{wang2022,kovalev2013quantum,panteleev2021degenerate,lin2024,lin2025single,koukoulekidis2401small,wang2023abelian,mostad2025advancing}.
Furthermore, the parameters of many small-size degenerate GB codes and their generalization to two block group algebra codes, as well as new distance bounds for non-degenerate codes, derived from their underlying algebraic structure, have been studied in \cite{wang2022,lin2024}. 
However, many of the currently studied GB codes are limited either to non-degenerate codes that require large qubit connectivity, or, in the case of degenerate codes, to those that have been studied primarily through numerical methods. 
One of the main remaining challenges is to systematically construct families of degenerate GB codes, ideally with fewer than ten qubit connectivity, that are suitable for practical applications. 

An efficient protocol recently has been proposed for implementing a class of GB codes on neutral atom array quantum computers, offering faster logical cycles and improved hardware efficiency compared to other approaches \cite{viszlai2023matching}. 
Other practical implementations of GB codes in neutral atom array and silicon spin qubit hardware have been investigated in \cite{viszlai2023matching, siegel2024towards}. 

Recently, in~\cite{wang2022,lin2024}, some distance bounds for GB codes, along with several numerical examples featuring small qubit connectivity, have been studied. 
In particular,~\cite[Statement~11]{wang2022} establishes a lower bound on the length of GB codes with qubit connectivity four and compares the tightness of this bound with that of the surface code family. 
Interestingly, it is claimed that no connectivity-four GB codes with parameters $[\![d^2, 2, d]\!]_2$ exist for even values of $d$ (except for the $[\![4, 2, 2]\!]_2$).

\subsection{Our main contributions}
In this work, we extend previous results on GB codes by establishing a connection with the family of additive cyclic codes and present the following new contributions.
\begin{itemize}
\item We present new general bounds (both lower and upper) on the minimum distance of GB codes.
We show that these bounds account for quantum code degeneracy and can yield tighter minimum distance estimates than those implied by classical bounds as the classical distance bounds cannot reach beyond the degeneracy.
We further demonstrate the application of these bounds to GB codes with qubit connectivity six, determining the true minimum distance in cases where other distance bounds fail.
\item We construct two families of highly degenerate GB codes with qubit connectivity four, having parameters $[\![d^2+1, 2, d]\!]_2$ for odd $d \geq 3$ and $[\![d^2, 2, d]\!]_2$ for even $d \geq 4$. 
The existence of the latter family was previously ruled out in~\cite{wang2022}.

\item We take advantage of our minimum distance bounds to form an infinite class of connectivity six GB codes with parameters
\[
[\![2(d+1)(2d+1),2,2(d+1)]\!]_2.
\]

\item We show how to perform a {\bf fault-tolerant transversal logical CNOT} on a large subset of GB codes with various qubit connectivity at zero cost using only a permutation of data qubits.  
Such 
a GB class enables noise-free logical CNOT operations by replacing in-block CNOTs with qubit permutations, which preserve the logical effect without introducing physical errors or propagating existing ones, even in the presence of non-Clifford gates. 
Compared to lattice surgery, they provide fully idealized entanglement without the time and error overhead associated with multiple rounds of physical operations, making them a powerful tool for efficient logical gate implementation.

\item We provide a syndrome extraction pattern for the connectivity four GB families that is resilient to hook errors introduced during the extraction process.
Additionally, we completely characterize the girth of the Tanner graphs for all GB codes.

\item We investigate the error thresholds of the two mentioned connectivity four GB code families under circuit-level noise by evaluating their logical and physical error rates using 
the belief propagation ordered statistics (BP-OSD)
and minimum-weight perfect-matching  (MWPM) decoders, obtaining thresholds comparable to those of rotated surface codes.
\end{itemize}
These results provide evidence supporting that the mentioned GB families could be promising candidates for fault-tolerant quantum computing. 
A summary of important features of our proposed connectivity four GB families are highlighted in Table \ref{tab:Gbproperties}

\begin{table*}[t]
\centering
\caption{Key properties of our two GB code families.}
\begin{tabular}{c|c|c}
\hline
\textbf{Property} 
& \textbf{$[\![d^2+1,2,d]\!]_2$ (odd $d$)} 
& \textbf{$[\![d^2,2,d]\!]_2$ (even $d$)} \\
\hline
Polynomial representation 
& $a(x)=1+x,\; b(x)=1+x^d$ 
& $a(x)=1+x,\; b(x)=1+x^{d+1}$ \\
\hline
Qubit connectivity 
& 4  
& 4 \\
\hline
Girth 
& 8 
& 8 \\
\hline
Permutation-type logical CNOT 
& Yes 
& -- \\
\hline
Effective distance 
& $d_{\mathrm{eff}} = d$ 
& $d_{\mathrm{eff}} = d$ \\
\hline
Decoding methods 
& MWPM, BP+OSD 
& MWPM, BP+OSD \\
\hline
Circuit-level threshold 
& $\approx 0.45\%$ ($\approx 1\%$ without idling noise) 
& $\approx 0.45\%$ ($\approx 1\%$ without idling noise)  \\
\hline
\end{tabular}
\label{tab:Gbproperties}
\end{table*}

Around the same time as our work, the existence of the two GB families of qubit connectivity four with the same parameters was established in \cite{arnault202522gbcodesclassificationcomparison, arnault2025generalization}, using a graph-theoretical approach independent of ours. 
The authors also studied equivalence of such codes and connection to other notable  CSS codes derived from Cayley graphs.
In contrast, our method for the existence of such families is based on the algebraic structure of GB codes, and the minimum distance bounds we derive can be applied to GB codes containing more than two logical qubits (with qubit connectivity more than four), a capability not attainable by the mentioned works.
Additionally, our work finds syndrome extraction patterns resistant to hook errors, and it studies the circuit-level performance of such families.

The structure of the paper is as follows.
In Section \ref{S:Intro} we give the preliminary background on classical linear codes, quantum codes, and GB codes. 
In Section \ref{S:dis.bound}, we give our minimum distance bounds for GB codes. 
In Sections \ref{S:oddfamily} and \ref{S:evenFamily} we study the existence of two families of optimal GB codes containing two logical qubits, logical operations, and a hook-error-resilient syndrome extraction. 
In Section \ref{sec: higher connectivity}, we discuss the construction of an infinite class of connectivity six GB codes and realization of transversal CNOT via permutation of data qubits.
The characterization of the girth of GB codes, and the results and discussion of our circuit-level noise simulations are studied in Section \ref{S:girth}.

\section{Preliminaries}\label{S:Intro}
In this section, we review the necessary background and establish the notation for classical and quantum error-correcting codes, which will serve as the foundation for presenting our results. 
Through this work, we only consider binary quantum error-correcting codes.

\subsection{Linear codes}

Let $\F_2$ be the binary field and  $n$ be a positive integer. 
A binary {\em linear code} of length $n$ is a linear subspace $C$ of $\F_2^n$. The {\em (Hamming) weight} of a vector $v \in \F_2^n$ is defined by the number of its non-zero coordinates and it is denoted by $\wt(v)$. 
The parameters of such a code are denoted by $[n,k,d]$, where $k$ is the dimension of the code $C$ (denoted by $k=\dim(C)$) and $d$ is the minimum distance of $C$ defined by
\[
d=d(C)=\min\{\wt(c):0 \neq c \in C\}.\]
A linear code $C\subseteq \F_2^n$ is called  \textit{cyclic} if for every $c=(c_0,c_1,\ldots,c_{n-1})\in C$, the vector $(c_{n-1},c_0,\ldots,c_{n-2})$ obtained by a cyclic shift of the coordinates of $c$ is also in $C$. 
In the study of cyclic codes, it is more convenient to represent vectors using their polynomial representation. 
In particular, for each $c=(c_0,c_1,\ldots,c_{n-1})\in \F_2^n$ the polynomial representation of $c$ is 
\[c(x)=\sum_{i=0}^{n-1}c_ix^i.\]
Throughout this paper, we frequently alternate
between the conventional binary representation
of a vector and its equivalent polynomial
representation.
It is well known that there is a one-to-one correspondence between cyclic codes of length $n$ over $\F_2$ and ideals of the ring $\F_2 [x]/\langle x^n-1\rangle $, for example see \cite[Section 4.2]{Huffman}. 
Under this correspondence, each cyclic code can be generated by a {\em unique monic polynomial} $g(x)$, where $g(x) \mid x^n-1$ and it has the minimal degree.
The polynomial $g(x)$ is called the {\em generator polynomial} of such cyclic code.

Recall that the Euclidean inner product of $u=(u_0,u_1,\ldots,u_{n-1})$ and $v=(v_0,v_1,\ldots,v_{n-1}) \in \F_2^n$ is defined as 
\[u\cdot v=\sum_{i=0}^{n-1}u_iv_i.\]
The Euclidean dual of a binary code $C$ of length $n$ is defined as 
$C^{\bot}=\{u\in\F_2^n: u\cdot c=0, \ \forall \ c \in C \}.$
Each linear code can be represented as the row space of a given matrix, called a {\em generator matrix}. In particular, if $C$ is a binary cyclic code of length $n$ and dimension $k$ with the generator polynomial $g(x)=\sum_{i=0}^{n-k}g_ix^i$, then the matrix
 \begin{equation}\label{generator matrix of linear cyclic codes}
 G=
 \begin{bmatrix}
g_0& g_1 & \cdots & g_{n-k}& 0 &0 & \cdots &0 \\
0 & g_0& g_1 & \cdots & g_{n-k}& 0 & \cdots &0 \\
 \vdots&\vdots & \vdots & \vdots &\vdots & \vdots & \vdots &\vdots \\
0 &0 &\cdots & 0 & g_0& g_1 & \cdots & g_{n-k}\\
\end{bmatrix}
\end{equation}
is a generator matrix of $C$.
An alternative approach to represent a linear code is by presenting it as a solution space of a matrix equation $Hx=0,x\in C$. The matrix $H$ is called a {\em parity check matrix} of such code, i.e. it defines the nullspace of matrix $G$. In particular, if $G$ and $H$ are generator and parity check matrices of a linear code, then we have $GH^T=0$.

\subsection{Quantum stabilizer codes}
Quantum stabilizer codes serve as the quantum analogue of classical linear error-correcting codes, providing protection against decoherence and other quantum noise sources \cite{Calderbank,Gottesman}.
In order to present them we need a number of preliminary definitions.

Let $(\mathbb{C}^2)^{\otimes n}$ be an $n$-qubits Hilbert space. 
A Pauli operator acting on an $n$-qubit system is an operator defined on $(\mathbb{C}^2)^{\otimes n}$ in the form of 
$$P=a \bigotimes_{i=1}^{n}P_i, $$ 
where $a \in \{\pm 1,\pm i\}$ and $P_i \in \{I,X,Y,Z\}$ are the Pauli matrices. 
The weight of a Pauli operator $P$ is the number of its non-identity components. The group of all Pauli operators on $n$-qubits $P_n$ is a non-commutative group. 
Its quotient  group  modulo the global phase factor $a$ is isomorphic to $\F_2^{2n}$ which is given by \cite{degen}
\begin{equation}\label{E:iso}
\begin{split}
\bigotimes_{i=1}^{n}P_i=&\bigotimes_{i=1}^{n}X^{x_i}Z^{z_i} \\&\rightarrow (x_1,x_2,\ldots,x_n|z_1,z_2,\ldots,z_n).
\end{split}
\end{equation}
Under this isomorphism, two Pauli operators  commute if and only if their binary representations, namely $(x|z)$ and $(x'|z')$, satisfy 
\begin{equation}\label{E:symplectic}
(x|z) \ast (x'|z')=x \cdot z' +x' \cdot z=0.
\end{equation}
The $\ast$ inner product is called {\em symplectic inner product} \cite{degen}.
A commutative subgroup $S$ of $P_n$ such
that $-I \not\in S$ is called a {\em stabilizer group}. 

Now we recall the definition of stabilizer codes. An $[\![n,k,d]\!]_2$ binary quantum stabilizer code is a $2^k$-dimensional subspace of $(\mathbb{C}^2)^{\otimes n}$ defined as the common eigenspace of a stabilizer group $S$ with eigenvalue $+1$ that has a minimal set of $n-k$ generators. The minimum distance $d$ is the minimum weight of a Pauli operator $P \in P_n$ such that $P$ commutes with all elements of $S$ but $P \not\in S$. 
A quantum code is called {\em degenerate} if there exists $P \in S$ with weight less than $d$. 
Degeneracy is a phenomenon unique to quantum codes, and highly degenerate codes, often characterized by low qubit connectivity, appear to offer improved error-correcting performance. In fact, the asymptotic results on the existence of good quantum LDPC codes and their finite length constructions are based on a degenerate construction \cite{panteleev2022asymptotically,panteleev2021degenerate}.

In order to construct a stabilizer group of $P_n$ one can take advantage of the isomorphism (\ref{E:iso}) and the commuting condition (\ref{E:symplectic}). In particular, any linear subspace $C$ of $\F_2^{2n}$ that satisfies (\ref{E:symplectic}) for any $(x|z)$ and $(x'|z') \in C$ is in correspondence to a stabilizer group. 
One special case of such construction of quantum stabilizer codes is called the Calderbank-Shor-Steane (CSS) that is given below \cite{CS96,steane1996multiple}. 

\begin{theorem}\label{T:CSS}
Let $C_2 \subseteq C_1$ be binary linear codes of length $n$ with dimensions $k_2$ and $k_1$, respectively. Then there exists an $[\![n,k=k_1-k_2,d]\!]_2$ binary quantum stabilizer code, where 
$$d=\min\{d(C_1 \setminus C_2), d(C_2^\bot \setminus C_1^\bot)\}.$$      
\end{theorem}
The parity check (stabilizer) matrix of such CSS code is 
\[
H=\begin{bmatrix}H_{C_2^\bot} & 0\\
0 & H_{C_1}
\end{bmatrix},\]
where $H_{C_2^\bot}$ and $H_{C_1}$ are parity check matrices of $C_2^\bot$ and $C_1$, respectively. 
In particular, we have $H_{C_2^\bot}H_{C_1}^T=0$ (recall that $C_2 \subseteq C_1$).

\subsection{Generalized bicycle codes}
In this subsection, we recall the construction of generalized bicycle (GB) codes.
These codes form a subclass of binary quantum CSS codes built from index-two quasi-cyclic codes with the aid of two circulant matrices \cite{kovalev2013quantum,panteleev2021degenerate,eczoo_generalized_bicycle}. 
We now briefly review the construction of these codes.

Let $a=(a_0,a_1,\ldots,a_{m-1})$ be a binary vector. 
Recall that the $m \times m$ circulant matrix corresponding to the vector $a$ is defined by 
\begin{equation}\label{E:circulant}
G_a=\begin{bmatrix}
a_0& a_{m-1} & a_{m-2} & \cdots & a_{1} \\
a_{1} & a_0& a_{m-1} & \cdots & a_{2}\\
 \vdots&\vdots & \vdots & \vdots &\vdots\\
a_{m-1} & a_{m-2}  & a_{m-3}&\cdots & a_{0}\\
\end{bmatrix}.
\end{equation}
Note also that there exists a ring isomorphism between $\F_2[x]/\langle x^m-1 \rangle$ and binary $m\times m$ circulant matrices given by
\[ \sum_{i=0}^{m-1}a_ix^i \rightarrow G_a, \]
where $a=(a_0,a_1,\ldots,a_{m-1})$.
From now on, we represent the matrix $G_a$ using its polynomial representation i.e., by $G_{a(x)}$. Some properties of circulant matrices are summarized below: 
\begin{itemize}
    \item $G_{a(x)b(x)}=G_{b(x)}G_{a(x)}$ (commuting).
    \item $G_{a(x)}^T=G_{a(x^{-1})}$.
    \item If $a(x)$ is invertible modulo $x^m-1$, then $G_{a(x)^{-1}}=(G_{a(x)})^{-1}$.
\end{itemize}
Let $G_{a(x)}$ and $G_{b(x)}$ be two $m \times m$ circulant matrices. 
Then the matrix
$$H=\begin{bmatrix}H_1 & 0\\
0 & H_2
\end{bmatrix},$$
where $H_1=[G_{a(x)}|G_{b(x)}]$ and $H_2=[G_{b(x^{-1})}|G_{a(x^{-1})}]$ is a stabilizer matrix of a quantum CSS code since 
\[
\begin{split}
H_1H_2^T&=[G_{a(x)}|G_{b(x)}][G_{b(x^{-1})}|G_{a(x^{-1})}]^T\\&=G_{a(x)b(x)}+G_{b(x)a(x)}=0
\end{split}
\]
Note also that, as it is shown in the proof of \cite[Proposition 1]{panteleev2021degenerate}, the matrices $H_1$ and $H_2$ both have rank $m-s$, where $$s=\deg(\gcd(a(x),b(x),x^m-1)).$$
Thus the quantum CSS code with the parity check matrix $H$ has parameters $[\![2m,2s]\!]_2$. 
The quantum codes constructed this way are known as {\em GB} codes.
A special subclass of GB codes is by choosing $a(x)=b(x)$, which is called a {\em bicycle} code \cite{mackay2004sparse}.

We highlight a connection between GB codes and quasi-cyclic codes. 
Circulant matrices arise from adding linear combinations of rows to the generator matrix of a cyclic code, reflecting the redundancy of minimum-weight stabilizers in GB codes \cite{lin2025single}. 
Horizontal concatenation of two circulants forms a quasi-cyclic generator matrix, making cyclic and quasi-cyclic code techniques useful for analyzing GB codes.

\section{New minimum distance bounds for GB codes}\label{S:dis.bound}
One of the key steps in constructing large families of GB codes is the development of tools for accurately computing their minimum distance. 
Previously, the only known bounds on the minimum distance of GB and two-block group algebra codes were given in~\cite{wang2022,lin2024}. 
In particular, the applicability of these bounds was limited to one of the following cases:
\begin{itemize}
\item Bounds derived from classical codes, which are restricted by the classical minimum distance and fail to capture code degeneracy.
\item Non-constructive existence results, similar to the Gilbert--Varshamov (GV) bound, which do not provide a practical method for determining the minimum distance of a given code.
\end{itemize}
In this work, we focus on degenerate GB codes with low qubit connectivity between data and check qubits. 
Consequently, the mentioned approaches are not well suited for our purposes. 
The lack of effective minimum distance bounds for such GB codes partly explains the limited number of known infinite families of degenerate GB codes.

Motivated by this, we develop new minimum distance bounds based on the algebraic structure of GB codes. 
These bounds enable the construction of novel families of degenerate GB codes and allow us to determine tight intervals for their minimum distance.

For the rest of this section, we present new minimum distance bounds for the GB codes. 
These bounds facilitate the exact calculation of the minimum distance for certain instances, significantly reduce the overall distance computation time, and support the construction of new GB code families. 
Important features of our distance bounds are that 
\begin{itemize}
    \item  they do not limit to non-degenerate GB codes, and can give bounds beyond the degeneracy of the code, and
    \item they can be applied to GB codes with any number of qubit connectivity.
\end{itemize}
In the following, we present our minimum distance bound.

\begin{theorem}\label{T:generaldistancGB}
Let $n$ be a positive integer, $f(x) \mid x^n-1$, and $p(x)\in \F_2[x]$ such that $\gcd(p(x),x^n-1)=1$. Let also 
\[
r(x)+p(x)s(x)\equiv \frac{x^n-1}{f(x)} \pmod{x^n-1}
\]
and 
\[
r'(x)+p(x)^{-1}s'(x)\equiv \frac{x^n-1}{f(x)} \pmod{x^n-1}
\]
for some $r(x),s(x), r'(x),s'(x) \in \F_2[x]$ and $d$ be the minimum distance of the length $n$ binary cyclic code generated by $\frac{x^n-1}{f(x)}$. Then the GB code corresponding to $(f(x),p(x)f(x))$ has parameters $[\![2n,2\deg(f(x)),d_{GB}]\!]_2$ where 
\[
\begin{split}
&\min\{\frac{(\wt(p(x))+\wt(p(x)^{-1}))d}{(m+\wt(p(x))\wt(p(x)^{-1}))},\frac{d}{m'}\} \le d_{GB} \le \\& \min\{\wt(p(x))+1,\wt(p(x)^{{-1}})+1 ,\\&\wt(s(x))+\wt(r(x))\},
\end{split}
\]
where $m=\max\{\wt(p(x)),\wt(p(x)^{-1})\}$ and 
\[
\begin{split}
m' \in \{&\max\{\wt(r(x)),\wt(s(x))\},\\&\max\{\wt(r'(x)),\wt(s'(x))\} \}.
\end{split}
\]
Moreover, there exist non-trivial logical operators of weights 
$\wt(p(x))+1$ and $\wt(s(x))+\wt(r(x)) \in C_1\setminus C_2$. 
\end{theorem}

\begin{proof}
We restrict our discussion to the GB code constructed from the polynomials $a(x)=f(x)$ and $b(x)=p(x)f(x)$. 
Then this GB code has parameters $[\![2n,2k]\!]_2$, where $k=\deg(f(x))$. 
Note that $f(x)$ is the generator polynomial of a length $n$ binary cyclic code and its Euclidean dual has generator polynomial $g(x)$, where $g(x)$ is the reciprocal polynomial of 
$h(x)=\frac{x^n-1}{f(x)}$ defined by $g(x)=x^{\deg(h(x))}h(x^{-1})$. 

This GB code is the CSS code of $C_2 \subseteq C_1$, where $C_2=\langle (a(x),b(x)) \rangle$ and $C_1=C_2+\Span\{(c,c')\}$, and $c,c' \in \langle g(x) \rangle$. 
One can also define $C_1$ as $C_1=\langle (b(x^{-1}),a(x^{-1}))\rangle^{\bot}$. Using the latter, each element of $C_1$ can be characterized as $(u(x),v(x))$ such that 
\begin{equation}\label{E:membershipgen}
\begin{split}
&u(x)b(x)+v(x)a(x)=u(x)f(x)p(x)+v(x)f(x) \\&\equiv f(x)(u(x)p(x)+v(x)) \equiv 0 \pmod{x^n-1}.
\end{split}
\end{equation}
This implies one of the following two cases
\[\textbf{(a)}\ \begin{split}
 u(x)p(x)+v(x) \equiv 0 \pmod{x^n-1} 
 \end{split}
 \]
or 
\[\textbf{(b)}\
\begin{split}
& u(x)p(x)+v(x) \not\equiv 0 \pmod{x^n-1} \\ &
\text{and} \ h(x)\mid u(x)p(x)+v(x).
\end{split}
\]
First we consider case (b). 
In this case $u(x)p(x)+v(x)$ is a non-zero element of the cyclic code generated by $h(x)$, which will be called $C_{h(x)}$. 
A similar calculation after multiplying sides of (\ref{E:membershipgen}) by $p(x)^{-1}$ shows that  $u(x)+p(x)^{-1}v(x) \in C_{h(x)}$.  Let $d=d(C_{h(x)})$. Then 
\begin{equation}\label{E:dgb1}
\begin{split}
d &\le \wt (u(x)p(x)+v(x))\\&\le \wt(p(x))\wt(u(x))+\wt(v(x)).
\end{split}
\end{equation}
and 
\begin{equation}\label{E:dgb2}
\begin{split}
d &\le \wt (u(x)+p(x)^{-1}v(x))\\&\le \wt(u(x))+\wt(p(x)^{-1})\wt(v(x)).
\end{split}
\end{equation}
An immediate application of this is that if $m=\max \{\wt(p(x)),\wt(p(x)^{-1})\}$, then we get 
\[
\begin{split}
&(\wt(p(x))+\wt(p(x)^{-1}))d \le\\& (m+\wt(p(x))\wt(p(x)^{-1}))\wt( (u(x),v(x))).
\end{split}
\]
Therefore, this GB  code has minimum distance $d_{GB}$, which satisfies
\begin{equation}\label{E:lowerbound(b)}
\frac{(\wt(p(x))+\wt(p(x)^{-1}))d}{(m+\wt(p(x))\wt(p(x)^{-1}))}\le d_{GB}.
\end{equation}

Now we consider case (a). First, let 
\begin{equation}\label{E:codeword}
r(x)+p(x)s(x)\equiv h(x) \pmod{x^n-1},
\end{equation}
for some $r(x)$ and $s(x)\in \F_2[x]$.
Then for each $u(x)$ such that $f\nmid u(x)$ we have 
\[
u(x)r(x)+u(x)p(x)s(x)\equiv u(x)h(x) \pmod{x^n-1},
\]
which is a nonzero codeword of $C_{h(x)}$. Therefore, 
\begin{equation}\label{E:dgb3}
\begin{split}
d &\le \wt (u(x)r(x)+u(x)p(x)s(x)\\&\le \wt(r(x))\wt(u(x))+\wt(s(x))\wt(u(x)p(x)).
\end{split}
\end{equation}
Hence if $m=\max\{\wt(r(x)),\wt(s(x))\}$, then 
\[
\frac{d}{m}\le d_{GB}.
\]
A similar argument using the polynomial $p(x)^{-1}$ results in
a new equation
\[
r'(x)+p(x)^{-1}s'(x)\equiv h(x) \pmod{x^n-1},
\]
and
if $m'=\max\{\wt(r'(x)),\wt(s'(x))\}$, then 
\[
\max\{\frac{d}{m},\frac{d}{m'}\}\le d_{GB}.
\]
Moreover, one can easily verify that $(1,p(x))$ and $(p(x)^{-1},1) \in C_1\setminus C_2$ and both satisfy $(a)$.
Hence one gets
\begin{equation}\label{E:dgb4}
\frac{d}{m'}\le d_{GB} \le \min\{\wt(p(x)), \wt(p(x)^{-1})\}+1.
\end{equation}
Note also that if $f\nmid s(x)$ in (\ref{E:codeword}), then $(s(x),p(x)s(x)+h(x))=(s(x),r(x)) \in C_1\setminus C_2$ and satisfies case (b). 
Therefore, we also have
\begin{equation}\label{E:dgb5}
d_{GB} \le \wt(s(x))+\wt(r(x)).
\end{equation}
This completes the proof.
$\hfill \square$
\end{proof}

It is important to emphasize that the above result is merely a special case of the distance lower and upper bounds given in (\ref{E:dgb1})--(\ref{E:dgb5}); in fact, tighter bounds may be obtained using these general equations.
Moreover, our discussion of the minimum distance does not rely entirely on the existence of $p(x)^{-1}$ and remains valid even when the condition $\gcd(p(x),x^n-1)$ is relaxed.
The following example illustrates an application of these observations and demonstrates how the above theorem can be used.

\begin{example}\label{Exa:1}
 Let $n=48$,  $f(x)=1+x+x^2$ (where $f(x)\mid x^n-1$), and $p(x)=1+x^3+x^6+\cdots+x^{18}$.
Note that $p(x)$ is invertible modulo $x^n-1$ with the inverse 
 \[
 \begin{split}
p(x)^{-1}&=x^3+x^6+x^{12}+x^{18}+x^{24}\\&+x^{27}+x^{33}+x^{39}+x^{45}.
 \end{split}
 \]
 Also, the cyclic code generated by $h(x)=\frac{x^n-1}{f(x)}$ has distance $\frac{2n}{3}=32$ and we have 
 \[
 \begin{split}
h(x)&=p(x)(x^6+x^7+x^{27}+x^{28})\\&+(1+x+x^3+x^4).
 \end{split}
 \]
 Applying the result of Theorem \ref{T:generaldistancGB} implies that the pair of polynomials $(f(x),f(x)p(x))$ forms a GB code with parameters $[\![96,4,d_{GB}]\!]_2$, where 
 \[
7<\min \{\frac{(9+7)32}{9+63}, \frac{32}{4} \}  \le d_{GB} \le \min \{8,10\},
 \]
 which implies that $d_{GB}=8$. 
\end{example}
It is important to note that the above code is degenerate (to $4$), and classical minimum distance bounds cannot provide a bound better than $4$. 
However, our distance bound surpasses this degeneracy and can even capture the true minimum distance of the code.
Another example of degenerate GB code where its true minimum distance calculation is obtained using our distance bound is given in Appendix \ref{Ap:s3}.

These examples show that our distance bound can surpass the degeneracy of a quantum code, providing rare upper and lower bounds,or even the true value, on the minimum distance of GB codes. 
In the next section, we use this bound to construct two families of dimension-two GB codes with four-qubit check connectivity and highlight their key properties. 
Importantly, the bound is general and can guide the design of GB codes with higher connectivity as well.

\section{Existence, logical operators, and hook errors in $[\![d^2+1,2,d]\!]_2$ GB code with odd $d$}\label{S:oddfamily}

The previous study of GB codes with qubit connectivity four is limited to establishing a lower bound on the code length~\cite[Statement~11]{wang2022}. 
In particular, it states that for a GB code with qubit connectivity four and parameters length $N$ and distance $d$, we have
\[
N \ge \begin{cases}
d^2 + 1 & \text{for odd } d,\\ 
d^2 & \text{for even } d.
\end{cases}
\]
In this and the following sections, we prove the reverse inequality using our previously established lower bound on the minimum distance of GB codes. 
This result establishes the existence of two optimal families of GB codes with parameters $[\![d^2+1,2,d]\!]_2$ for odd $d \ge 3$ and $[\![d^2,2,d]\!]_2$ for even $d \ge 4$.

The novelty of these GB families is evident from the following observations:
\begin{itemize}
    \item For odd $d$, the numerical examples of GB codes with qubit connectivity four reported in~\cite{wang2022} do not consistently follow the pattern $[\![d^2+1,2,d]\!]_2$ (e.g., for $d = 7, 9, 13, 15$, etc.).
    \item In that work, the authors also claimed that no GB code with qubit connectivity four and parameters $[\![d^2,2,d]\!]_2$ exists for even values of $d$ (except for the $[\![4,2,2]\!]_2$ case).
\end{itemize}

These results will shed light on the structure of GB codes and open up new perspectives for systematically constructing other families of GB codes.
As we mentioned earlier, these GB families were also discussed in \cite{arnault202522gbcodesclassificationcomparison,arnault2025generalization}, around the same time as us, using a graph-theoretical approach.
However, our approach is different, and the following subsection reveals some new applications of our algebraic perspective.
 
A quantum code with stabilizer $H$ is $(w_c,w_r)$-regular if each column and row of $H$ has weight $w_c$ and $w_r$, respectively. 
Smaller weights are preferred, as fewer gates are needed for syndrome extraction, reducing error accumulation from noisy operations. 
Recall also that all weight-two polynomials of length $n$ belong to the cyclic code $\langle 1-x \rangle$.

In general, GB codes constructed from two polynomials of weight two can have different dimensions. In this section, we only consider the case when the GB code corresponding to two polynomials $a(x)$ and $b(x)$ has dimension two or equivalently when $\langle a(x),b(x) \rangle =\langle (1-x)\rangle$. 
We define $P_n(x)=\frac{x^n-1}{x-1}=1+x+\cdots+x^{n-1}$ over $\F_2[x]$.

\begin{lemma}\label{L:Pd inverse}
Let $d\ge$ be a positive odd integer and $n=\frac{d^2+1}{2}$. 
\begin{enumerate}
    \item The inverse of the polynomial $P_d(x)$ modulo $x^n-1$ and modulo $P_n(x)$ is $xP_d(x^d)$.  In particular, the polynomial $P_d^{-1}(x)$, in both cases, has weight $d$.
    \item We have
    \[
\begin{split}
P_n&=P_d(x)(1+x^d+x^{2d}+\cdots+x^{d(\frac{d-1}{2}-1)})\\&+(x^{d\frac{d-1}{2}}+x^{d\frac{d-1}{2}+1}+\cdots+x^{n-1}).
\end{split}
\]
    \end{enumerate}
\end{lemma}
\begin{proof}
The proof is given in Appendix \ref{Ap:s3}.
$\hfill \square$ 
\end{proof} 

Now we give an algebraic proof for the existence of the GB family $[\![d^2+1,2,d]\!]_2$ for each odd integer $d \ge 3$.

\begin{theorem}\label{T:QGB}
Let $d\ge 3$ be an odd integer. Then there exists a family of $[\![d^2+1,2,d]\!]_2$ which is $(2,4)$-regular. 
Except for $d=3$, all such codes are degenerate to 4.
\end{theorem}

\begin{proof}
Let $d\ge 3$ be a positive odd integer and $n=\frac{d^2+1}{2}$. Fix 
$f(x)=1+x$ and $f(x)P_d(x)$ as the corresponding polynomials of the GB code. Note that using Lemma \ref{L:Pd inverse} we have: 
\begin{equation}\label{E:minword1}
\begin{split}
P_n&=P_d(x)s(x)+r(x),
\end{split}
\end{equation}
where $s(x)=(1+x^d+x^{2d}+\cdots+x^{d(\frac{d-1}{2}-1)})$ and $r(x)=(x^{d\frac{d-1}{2}}+x^{d\frac{d-1}{2}+1}+\cdots+x^{n-1})$.
To compute the parameters of such GB code, we use the result of Theorem \ref{T:generaldistancGB}. In particular such GB code has parameters $[\![2n,2,d_{GB}]\!]_2$, where 
$$ \min\{\frac{2n}{d+1},\frac{2n}{d+1}\} \le d_{GB} \le \min \{d+1,d\}.$$
Here we used the fact that the cyclic code generated by $P_n(x)=\frac{x^n-1}{x-1}$ has distance $n$ (repetition code of length $n$). 
Hence we get
$$ d \le d_{GB} \le \min \{d+1,d\},$$
which implies that $d_{GB}=d$.
Moreover, this code has logical operators of weights $d$ and $d+1$.

As all the stabilizer
generators of such GB codes have weight four. Hence the
degeneracy claim follows immediately.
$\hfill \square$ \end{proof} 

Later, we prove that all such $[\![d^2+1,2,d]\!]_2$ GB codes have girth $8$ except when $d=3$, where the girth is $6$.

\subsection{Logical Pauli and CNOT operators of $[\![d^2+1,2,d]\!]_2$ GB codes} \label{S:logical}

Fault-tolerant logical gates can be implemented via permutations of the code space \cite{grassl2013leveraging,sayginel2024fault}, which relabel data qubits without generating new entanglement and can achieve zero overhead with perfect fidelity when incorporated at the compilation level. 
Recent works have focused on quantum codes with this property~\cite{chakraborty2026no,koh2026entangling}; for example, \cite{koh2026entangling} introduced \emph{phantom codes}, allowing logical CNOT gates between any logical qubits using only qubit permutations.

In standard Clifford group circuits, entangling gates such as CNOT can often be avoided by updating a classical record of Pauli operators, a technique known as Pauli frame tracking \cite{chamberland2018fault}. 
However, this approach can fail in the presence of non-Clifford (magic) gates, where physical CNOT operations are generally required. 
In contrast, phantom codes overcome this limitation by replacing in-block CNOT gates with qubit permutations. 
These permutations implement the same logical effect as CNOTs but do not correspond to physical operations and therefore introduce no additional noise. 
Therefore, {\em no new physical errors are generated} and {\em errors are not propagated} in contrast to real CNOT gates.

In this section, we prove that the GB family $[\![d^2+1,2,d]\!]_2$ has the property that logical CNOT gate can be realized through a permutation of data qubits, for each odd $d\ge 3$. We also find the structure of minimum weight logical Pauli operators. 

Let $d$ be a positive odd integer, and let $C_2$ denote the binary code of length $d^2+1$ generated by the cyclic shifts of $(1+x,1+x^d)$. 
This procedure
forms a quasi-cyclic code, of index two, that is in
correspondence to the $X$-stabilizers of the GB family $[\![d^2+1,2,d]\!]_2$ discussed in the previous section. 
Furthermore, note that $C_1 = C_2 \oplus \Span\{(P_n(x),0),(0,P_n(x))\}$, recall that $P_n(x)=1+x+\cdots+x^{n-1}$, forms the
set of all codewords corresponding to the
$X$-normalizers; 
that is, the GB family is the CSS code  of $C_2 \subseteq C_1$.

Let $D$ be the $2 \times 2n$ matrix with the rows $v_1=[1,1,\ldots,1,0,0,\ldots,0]$ and $v_2=[0,0,\ldots,0,1,1,\ldots,1]$, then for each $u \in \F_2^2$, the encoded logical state corresponding to $u$ in such quantum codes, up to a normalizer constant, is defined by 
\begin{equation}\label{E:encoding}
\ket{u}_L=\sum_{c \in C_2}\ket{c+(uD)}.
\end{equation}
In particular, the four logical states for each code inside the family are \\
$\ket{00}_L=\displaystyle\sum_{c \in C_2}\ket{c}$, $\ket{10}_L=\displaystyle\sum_{c \in C_2}\ket{c+v_1}$, $\ket{01}_L=\displaystyle\sum_{c \in C_2}\ket{c+v_2}$, and $\ket{11}_L=\displaystyle\sum_{c \in C_2}\ket{c+v_1+v_2}$.
Thus, the logical operators $XI$, $IX$, and $XX$ (similarly, $ZI$, $IZ$, and $ZZ$) can be implemented by acting on $n$ or $2n$ qubits. 
However, one can identify more optimal implementations that require interaction with fewer data qubits. 
In the following theorem, we discuss such choices based on the arguments presented in the proof of Theorem \ref{T:QGB}.

\begin{theorem}\label{T: Logicals}
Consider the $[\![d^2+1,2,d]\!]_2$ GB family, where $d$ is an odd integer. One can perform logical $X$ operators as follows.
\begin{itemize}
    \item $XI$ by performing $(u(x),v(x))$ (or any cyclic shift of it), 
    \item $XX$ by performing $(xv(x^d),u(x^d))$ (or any cyclic shift of it),
    \item $IX$ by performing $(1,P_d(x))$ or $(P_d(x)^{-1},1)$ (or any cyclic shift of them), 
\end{itemize}
where $$u(x)=1+x^d+x^{2d}+\cdots+x^{(\frac{d-1}{2}-1)d}$$ and $$v(x)=x^{n-\frac{d+1}{2}}+x^{n-\frac{d+1}{2}+1}+\cdots+x^{n-1}.$$
The first two operators act on $d$ data qubits, while the last one acts on $d+1$ qubits; each is optimal in this sense.
\end{theorem}
\begin{proof}
The proof is given in Appendix \ref{Ap:s3}.
$\hfill \square$ 
\end{proof} 

Note that so far we have only talked about logical $X$-operators above. 
Since each $X$-stabilizer (or $X$-normalizer) of the form $(a(x), b(x))$ is equivalent, via a one-to-one mapping, to $(b(x^{-1}), a(x^{-1}))$, which is a $Z$-stabilizer (or $Z$-normalizer), one can apply the result of Theorem \ref{T: Logicals} to obtain minimum-weight candidates for implementing the logical operators $ZI$, $IZ$, and $ZZ$.

Next we show that logical CNOT operator, i.e., the CNOT gate between the two logical qubits in a single GB code, can be done at ``zero cost'' by solely permuting (relabelling) the data qubits. Such permutation  technique previously studied in the literature \cite{grassl2013leveraging, sayginel2024fault} for general quantum codes and it is distinct from other fault-tolerant approaches such as transversal gates and lattice surgery.
Recall that in the $[\![d^2+1,2,d]\!]_2$ GB family, each data qubit has a label from 
\[\mathbb{Z}_n \times \mathbb{Z}_n=\{0,1,\ldots,n-1\}\times \{0,1,\ldots,n-1\}.\]
So each vector in $\F_2^n \times \F_2^n$ can be expressed by $(a(x),b(x))\in \F_2[x]/\langle x^n-1 \rangle \times \F_2[x]/\langle x^n-1 \rangle$.
We will take advantage of the following permutations in the next theorem: 
\[
\begin{split}
S\big((a(x),b(x))\big)\rightarrow (b(x),a(x)),
\end{split}
\]
\[
\begin{split}
E_d\big((a(x),b(x))\big)\rightarrow (a(x^d),b(x^d)),
\end{split}
\]
and
\[
\begin{split}
\Pi_{1,0}
\big((a(x),b(x))\big)\rightarrow (xa(x),b(x)),
\end{split}
\]
where the first permutation swaps the first block of size $n$ with the second size $n$ block, the second permutation
sends the element in coordinate $i$ to $(d\times i) \mod n$, and the third one applies a cyclic shift to the first block.
In the next theorem, we choose the basis of logical $X$ operators as $(P_n,0)$ for $XI$, and $(P_n,P_n)$ for $IX$.

\begin{theorem}\label{T:CNOT action}
Let $d$ be a positive odd integer. For each GB code inside the family $[\![d^2+1,2,d]\!]_2$, the logical CNOT operator can be implemented by applying the permutation $\Pi_{1,0}E_dS$.     
\end{theorem}

\begin{proof}
Let $O=\Pi_{1,0}E_dS$. 
It acts as   
\[
O\big((u(x),v(x))\big)=(xv(x^d),u(x^d)).
\]
The operator $O$ maps each $X$-stabilizer generator $(x^i(1+x),x^i(1+x^d))$ to another stabilizer generator namely 
$(x^{id+1}(1+x^{d^2}),x^{id}(1+x^d))$, where, using the fact that $d^2\equiv -1 \pmod n$, it can be equivalently represented as $(x^{id}(1+x),x^{id}(1+x^d))$ which is another $X$-stabilizer. Since $\gcd(n,d)=1$, this permutation gives a one-to-one mapping of the stabilizer generators. 
Therefore, using (\ref{E:encoding}) and Theorem \ref{T: Logicals}, one can verify that $O$, at the logical level, maps
\begin{itemize}
    \item $\ket{00}_L\rightarrow\ket{00}_L$.
    \item $\ket{01}_L \rightarrow \ket{01}_L$ because applying $O$ to $(1,P_d(x))$ gives $(P_d(x)^{-1},1)$, see also Lemma \ref{L:Pd inverse} part (1).
    \item $\ket{10}_L \leftrightarrow \ket{11}_L$ as it swaps the first two operators given in Theorem \ref{T: Logicals}.
\end{itemize}
Hence $O$ acts as CNOT  at the logical level for each code inside the family $[\![d^2+1,2,d]\!]_2$.  
$\hfill \square$ \end{proof} 

A method of implementing logical CNOT gates in copies of rotated surface codes is lattice surgery,
which is geometrically local and experimentally realistic, making it well suited to certain hardware architectures \cite{horsman2012surface}.
Compared to lattice surgery, our GB family achieves {\em zero-cost logical CNOT operations}, as these are realized via qubit permutations rather than physical gates, thereby introducing neither additional noise nor error propagation, and incurring no time overhead. 
In contrast, lattice surgery requires $O(d)$ rounds of stabilizer measurements and physical operations, resulting in a finite time cost and the accumulation of errors.

Overall, our GB family offers idealized, noise-free entanglement, whereas lattice surgery provides a practical but resource-costly implementation of logical gates under realistic constraints. 
As noted earlier, another application of such FT CNOT gates is that they can serve as a replacement for Pauli frame tracking in the presence of non-Clifford gates.

\subsection{Hook error and the effective weight}\label{Sec:hook}

Correlated faults during syndrome extraction, known as hook errors, can reduce the effective distance of a quantum code by propagating a single check-qubit error to multiple data qubits, potentially forming low-weight logical errors \cite{fowler2012surface,tomita2014low,dennis2002topological, geher2024error, manes2025distance}.

Following the approach of \cite{manes2025distance}, we define {\em effective distance} of a quantum code under a given measurement circuit to be the minimum number
of physical (data and check) errors required for an undetectable logical error.
For the GB family $[\![d^2+1,2,d]\!]_2$, we analyze the impact of such errors under a specific syndrome extraction scheme. 
In particular, using a Right-Left (RL) extraction order (shown in Figure 2 of Appendix \ref{Ap:s3}), all hook errors take a restricted form that does not align with minimum-weight logical operators. 
As a result, the effective distance remains equal to the theoretical distance.

\begin{theorem}\label{T:hook1}
For odd $d \ge 3$, the GB family $[\![d^2+1,2,d]\!]_2$ achieves effective distance $d$ under the RL syndrome extraction scheme.
\end{theorem}

This demonstrates that, with an appropriate extraction order, hook errors do not degrade the code performance. In contrast, alternative extraction patterns can significantly reduce the effective distance, in some cases to approximately $\lceil d/2 \rceil$. 
Full proofs and additional constructions are provided in Appendix \ref{Ap:s3}.

\section{$[\![d^2,2,d]\!]_2$ GB Family with even $d$}\label{S:evenFamily}
As we mentioned earlier, the existence of an even-distance GB code family of the form $[\![d^2, 2, d]\!]_2$ has been overlooked \cite{wang2022}, and it is claimed that there is no family of GB codes with these parameters, except the $[\![4, 2, 2]\!]_2$. 

In this section, we follow similar steps to those in the previous section to construct this family. 
First, we need the following lemma. 
Recall that $P_n(x)=1+x+\cdots+x^{n-1}$.

\begin{lemma}
Let $d$ be a positive even integer and $n=\frac{d^2}{2}$. Then the following holds.
\begin{enumerate}
    \item $P_{d+1}(x)$ is invertible and 
\[
\begin{split}
&P_{d+1}(x)^{-1}\equiv
x(1+x^{d+1}+x^{2(d+1)}+\cdots\\&+x^{(\frac{d}{2}-1)(d+1)})+x^{\frac{d}{2}+1}(1+x^{d+1}+x^{2(d+1)}+\cdots\\&+x^{(\frac{d}{2}-2)(d+1)})
\pmod{x^n-1}.
\end{split}
\] 
Moreover, $\wt(P_{d+1}(x)^{-1})=d-1$.
\item We have
\[
\begin{split}
P_n&=P_{d+1}(x)(x^{\frac{d}{2}+1})(1+x^{d+1}+x^{2(d+1)}+\cdots\\&+x^{(\frac{d}{2}-2)(d+1)})
+(1+x+x^{2}+\cdots+x^{\frac{d}{2}})
\end{split}
\]
and 
\[
\begin{split}
P_n&=P_{d+1}(x)^{-1}(1+x+x^2+\cdots+x^{(\frac{d}{2}-1)})\\&
+(1+x^{d+1}+x^{2(d+1)}+\cdots+x^{(\frac{d}{2}-1)(d+1)}).
\end{split}
\]
\end{enumerate}
\end{lemma}

\begin{proof}
The proof follows from a straightforward calculation modulo $x^n-1$.    
$\hfill \square$ \end{proof} 

Next we use the result of Theorem \ref{T:generaldistancGB} to prove the existence of $[\![d^2,2,d]\!]_2$ family of GB code, where all the stabilizer generators have weight four. 

\begin{theorem}\label{T:QGB2}
Let $d\ge 4$ be an even integer. Then there exists a family of $[\![d^2,2,d]\!]_2$ which is $(2,4)$-regular. All such codes are degenerate to 4, except when $d=4$.
\end{theorem}
\begin{proof}
The proof is similar to that of Theorem \ref{T:QGB} and it is given in Appendix \ref{Ap:s3}
$\hfill \square$ 
\end{proof}

In the previous section, Theorem \ref{T:hook1}, we showed that the RL syndrome extraction pattern implies the effective minimum distance equal to the theoretical minimum distance for the $[\![d^2+1,2,d]\!]_2$ GB family. 
The following theorem shows the same outcome for the GB family discussed in this section.

\begin{theorem}\label{T:hook2}
Let $d\ge 4$ be an even integer. 
Then the $[\![d^2,2,d]\!]_2$ GB family has effective distance $d$ considering the RL pattern syndrome extraction.     
\end{theorem}

\begin{proof}
The proof follows from an argument similar to that of Theorem \ref{T:hook1}.    
$\hfill \square$ \end{proof} 

It is worth noting that low-weight logical $X$ and $Z$ operators for this family of GB codes can also be found using the approach described in Section~\ref{S:logical}. 
Therefore, we omit such discussions here.

\section{Beyond the connectivity four families}\label{sec: higher connectivity}
In this section, we leverage the results developed in the previous sections to move beyond the previously studied connectivity four families. 
Recall that the key tool enabling tight minimum distance bounds and permutation-type logical CNOT operators was the structure of the polynomial $p(x)$, where the corresponding GB codes were defined via $a(x)$ and $b(x)=a(x)p(x)$. 
In particular, $p(x)^{-1}$ was a permutation of $p(x)$.

Motivated by this property, we investigate polynomials $p(x)$ exhibiting this feature, and then utilize them to construct GB codes with desirable logical operators obtained through data-qubit permutations.

We are interested in polynomials $p(x) \in \F_2[x]$ such that $p(x)p(x^r) \equiv x^\ell \pmod{x^n-1}$ 
for positive integers $n, r, \ell$, to construct GB codes with generator polynomials $a(x)$ and $b(x)=a(x)p(x)$. 
A polynomial $p(x)$ satisfying this condition will be called {\em nearly self-inverse}, or {\em self-inverse} if $\ell=0$. 
We classify the structure of such polynomials for odd values of $n$, and present a general form for even values of $n$ in Appendix \ref{Ap:s4}. 
The first application of this selection is presented below.

\begin{theorem}\label{T:distance bound derivation}
Let $n$ be a positive integer and $p(x) \in \F_2[x]$ such that $p(x)p(x^r) \equiv x^\ell \pmod{x^n-1}$ for some integers $2\le r \le n-1$ and $\ell$. 
If 
\[
r(x)+p(x)s(x)\equiv \frac{x^n-1}{f(x)} \pmod{x^n-1}
\]
for some polynomial $f(x)$ such that $f(x) \mid x^n-1$, then the GB code corresponding to $(f(x),f(x)p(x))$ satisfies the following distance bound
\[
\begin{split}
&\min\{\frac{2d}{(1+\wt(p(x)))},\frac{2d}{\wt(r(x))+\wt(s(x))}\} \le d_{GB} \le \\& \min\{\wt(p(x))+1 ,\wt(s(x))+\wt(r(x))\},
\end{split}
\]
where $d$ is the minimum distance of the dual of cyclic code generated by $f(x)$.
\end{theorem}
This derivation follows directly from the result of Theorem \ref{T:generaldistancGB} and its proof.

\begin{example}[Qubit connectivity six]
(1) Let $n=26$, $a(x)=1+x$, and $b(x)=1+x^3+x^{13}+x^{17}$. An easy calculation shows that $p(x)=x^{16} + x^{15} + x^{14} + x^{13} + x^2 + x + 1$ and $p(x)p(x)=x^6 \pmod{x^n-1}$. Moreover,
\[
\begin{split}
&(x^{10} + x^6 + x^3 + 1)p(x)+ (x^{19} + x^{16} + x^9 + 1)\\& =P_n(x) \pmod{x^n-1}
\end{split}
\] 
Now the minimum distance bound of the above theorem implies that 
\[
7 \le d_{GB} \le 8.
\]
One can show that when $n$ is even, the GB codes of length $n$ always have even weight logical operators (see the proof of the next theorem). Therefore, the GB code of $a(x)$ and $b(x)$ has parameters $[\![52,2,8]\!]_2$.
\end{example}

We take advantage of the result of Theorem \ref{T:distance bound derivation} to build another family of GB codes with connectivity six and better rate than the previous connectivity four families. 
\begin{theorem}
Let $d\ge 3$ be an odd integer. Then
there exists a family of connectivity six GB codes with parameters
\[
[\![2(d+1)(2d+1),2,2(d+1)]\!]_2.
\]
\end{theorem}
\begin{proof}
The proof follows by fixing $n=(d+1)(2d+1)$ and  $p(x)=P_d(x)+x^{\frac{n}{2}}P_{d+1}(x)$. 
Then one can show that $p(x)$ satisfies the conditions of Theorem \ref{T:distance bound derivation} and the GB code of $a(x)=1+x$ and $b(x)=a(x)p(x)$ has the desired parameters. More details are provided in the Appendix \ref{Ap:s4}. 
$\hfill \square$
\end{proof}
The first few codes in this family are $[\![56,2,8]\!]_2$, $[\![132,2,12]\!]_2$, and $[\![240,2,16]\!]_2$. 

A key feature shared by the GB families considered in this work is the presence of a nearly self-inverse polynomial $p(x)$. 
In Appendix \ref{Ap:s4}, we classify all self-inverse polynomials for any odd $n$. Below, we present an application of these polynomials.

For the remainder of this section, we focus on odd-length GB codes, noting that similar results may hold for even-length codes as well.
\begin{theorem}
Let $n$ be an odd integer and  $p(x)$ be a self-inverse polynomial modulo $x^n-1$ i.e., $p(x)p(x^r)\equiv 1 \pmod{x^n-1}$. 
\begin{itemize}
    \item One can characterize all such polynomial through the process explained in Appendix \ref{Ap:s4}.
    \item If $f(x)\mid x^n-1$, $\langle f(x)\rangle=\langle f(x^r) \rangle$ (they generate the same cyclic code), and $r^2\equiv 1 \pmod{n}$,  then the GB code of $a(x)=f(x)$ and $b(x)=f(x)p(x)$, which has parameters $[\![2n,2\deg(f(x)),d]\!]_2$, has a transversal logical CNOT gate (i.e., $\otimes_{i=1}^{\deg(f(x))} \text{CNOT}$) realizable through a permutation of data qubits.
\end{itemize}

\end{theorem}
The proof along with a short list of examples are provided in Appendix \ref{Ap:s4}.

\section{Girth and Circuit-Level Noise Thresholds of GB codes}\label{S:girth}
\subsection{Girth of GB Codes}

The girth of a Tanner graph plays a crucial role in the performance of message-passing decoders, as short cycles introduce correlations that degrade decoding accuracy. 

For GB codes, the girth is determined by the algebraic structure of the defining polynomials. 
In particular, one can explicitly characterize when short cycles arise. 
Conditions for girth $4$ and $6$ correspond to specific alignments of nonzero coefficients in the defining polynomials $a(x)$ and $b(x)$, which create short cycles in the associated Tanner graph.

Moreover, the girth of GB codes is fundamentally bounded. 
In particular, for any GB code defined by polynomials of weight at least two, the Tanner graph satisfies
\[
g(G) \leq 8.
\]
Thus, the maximum achievable girth of a GB code is eight. Indeed the GB codes of Sections \ref{S:oddfamily} and \ref{S:evenFamily} have girth eight.

These results provide useful design criteria: avoiding the algebraic configurations that induce girth $4$ and $6$ cycles leads to Tanner graphs with improved decoding performance.
A complete characterization of the girth conditions, together with all proofs, is deferred to Appendix \ref{Ap:S5}.

\subsection{Circuit-level noise simulation}

To evaluate the practical performance of GB codes, we perform circuit-level noise simulations for the families introduced in Sections~\ref{S:oddfamily} and \ref{S:evenFamily}.
We consider odd-distance codes $[\![d^2+1,2,d]\!]_2$ for $d \in \{3,7,11,15\}$ and even-distance codes $[\![d^2,2,d]\!]_2$ for $d \in \{4,8,12,16\}$, and compare their performance to rotated surface codes.

We adopt a standard circuit-level depolarizing noise model with physical error rate $p$, where errors occur independently after every operation, including state preparation, measurement, gates, and idle periods. 
The syndrome extraction circuits are designed based on the RL pattern discussed earlier to avoid hook errors, ensuring that correlated faults do not reduce the effective distance.
The simulation was designed using the open-source package \texttt{qLDPC} \cite{perlin2023qldpc}, where the decoders are based on the packages \texttt{ldpc} \cite{Roffe_LDPC_Python_tools_2022} and \texttt{Pymatching} \cite{higgott2021pymatchingpythonpackagedecoding}.

For each code and physical error rate, we perform $d$ rounds of syndrome extraction and estimate logical error rates via Monte Carlo simulations using MWPM and BP-OSD decoders. Owing to the CSS structure, $X$- and $Z$-errors are decoded independently, reducing the problem to two classical decoding tasks.

\begin{figure}[htbp]
\includegraphics[width=0.47\textwidth]{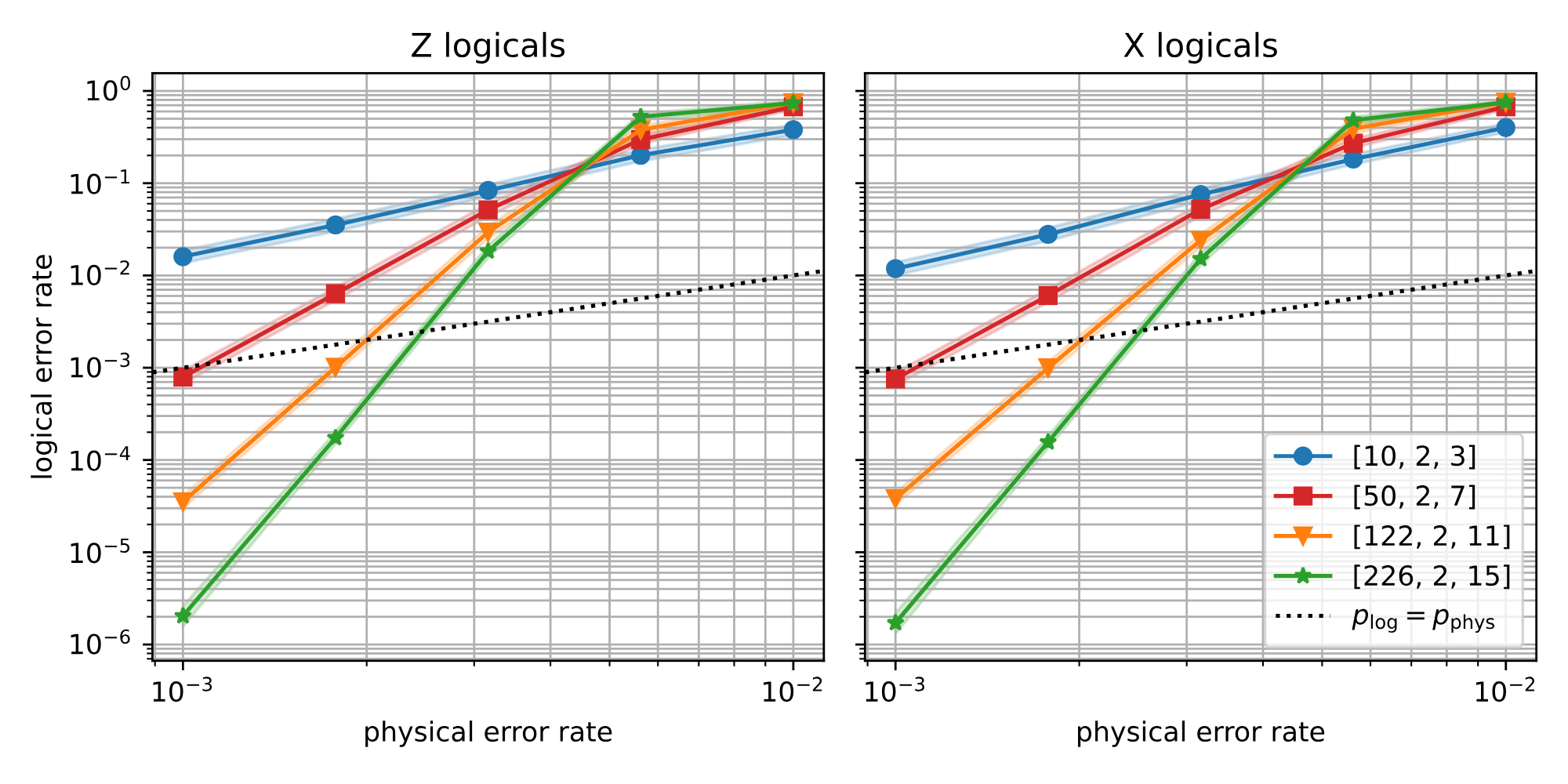}
    \caption{
        Logical versus physical error rate for $[\![d^2+1,2,d]\!]_2$ GB family and odd $d\in \{3,7,11,15\}$ using MWPM under circuit-level depolarizing noise.}
    \label{fig:logical_error_rates}
\end{figure}

\textbf{Threshold.} Due to structural differences between the parity-check matrices $H_X$ and $H_Z$, the decoding performance for $X$- and $Z$-type errors differs, and we therefore report their thresholds separately. The overall threshold is defined as the minimum of the two.

The simulation results, shown in Figures~\ref{fig:logical_error_rates} and \ref{fig:logical_error_rates2}, exhibit a clear threshold behavior under circuit-level depolarizing noise. 
For both GB code families, we observe an overall threshold of approximately $0.45\%$ using MWPM decoding. 
The diagram for BP-OSD is omitted due to the similarity in behavior.

Below the threshold, logical error rates decrease with increasing code distance, demonstrating effective fault-tolerant error suppression. 
Moreover, the observed thresholds are comparable to those of rotated surface codes, which typically lie in the range $0.5\%$--$1.0\%$ under similar noise models \cite{fowler2012surface,demarti2024decoding,higgott2023improved,higgott2025sparse}.

\begin{figure}[htbp]
\includegraphics[width=0.47\textwidth]{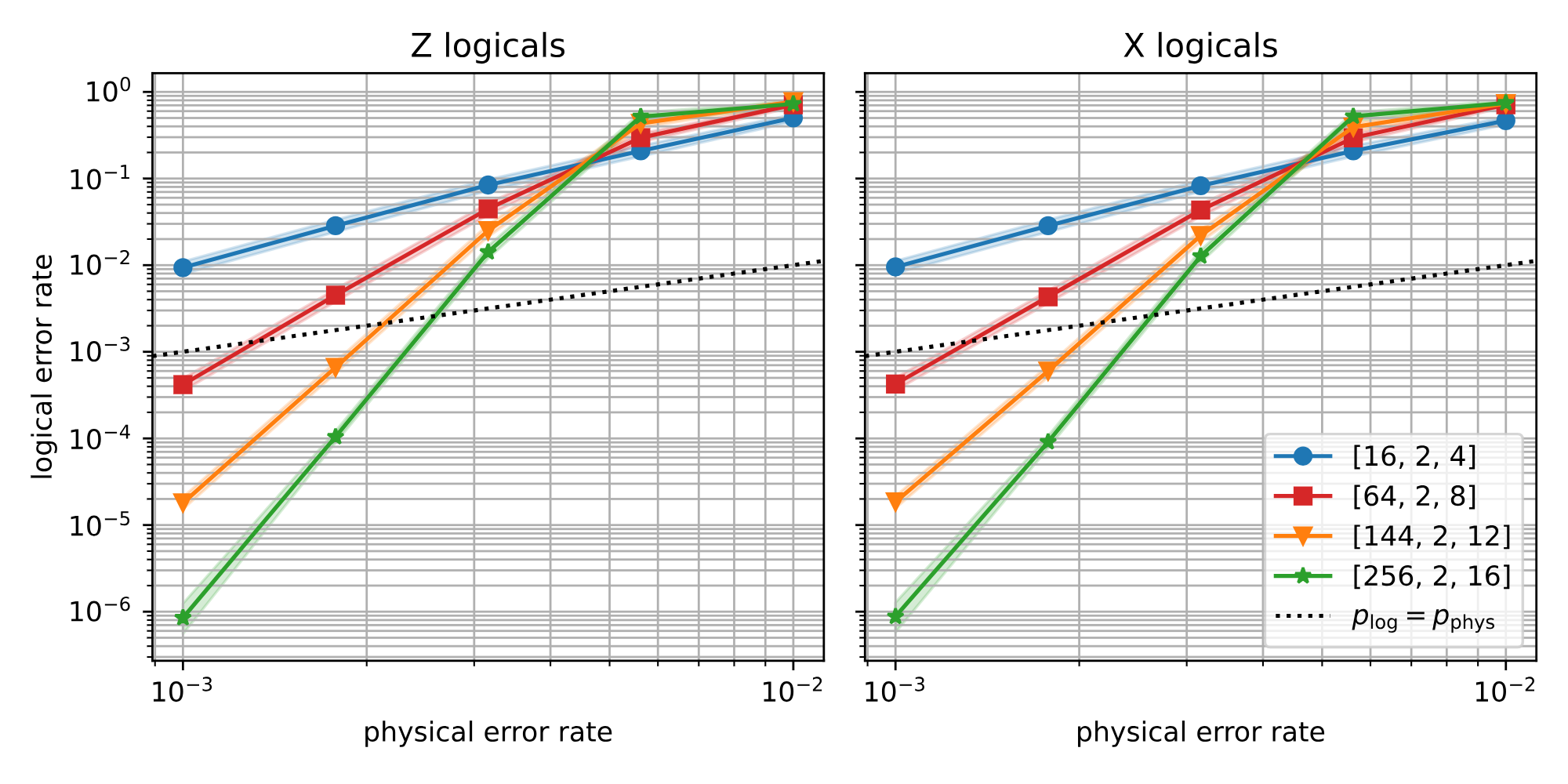}
    \caption{
        Logical versus physical error rate for the $[\![d^2,2,d]\!]_2$ GB family with $d \in \{4,8,12,16\}$, using the MWPM decoder under circuit-level depolarizing noise.}
    \label{fig:logical_error_rates2}
\end{figure}

We also performed circuit-level simulations for the GB code families under MWPM with idling errors removed and $10^{10}$ shots. 
This led to an improved performance, yielding a threshold of approximately $1\%$ for both odd and even distance families. See Appendix \ref{Ap:S6}. 

The results demonstrate that GB codes achieve threshold performance comparable to rotated surface codes, while offering additional advantages such as higher encoding rates, a redundant stabilizer, and a regular structure conducive to efficient hardware implementation.

\section{Conclusion and future directions}
In this work, we established a connection between GB codes and additive cyclic codes over $\mathbb{F}_4$ (see SM), resulting in tight minimum distance bounds beyond the degeneracy limit for GB codes. 
Using these bounds, we identified two families of highly degenerate GB codes, $[\![d^2+1,2,d]\!]_2$ for odd $d \ge 3$ and $[\![d^2,2,d]\!]_2$ for even $d \ge 4$, both with qubit connectivity four, and a connectivity six family with parameters
$[\![2(d+1)(2d+1),2,2(d+1)]\!]_2$ for odd $d \ge 3$.
For the connectivity four families, we realized an FT logical CNOT, proposed a syndrome extraction scheme preserving effective distance, and characterized attainable girths by GB codes. 
Circuit-level simulations under depolarizing noise with BP-OSD and MWPM decoders show threshold behavior comparable to rotated surface codes. 
Finally, we proposed an FT implementation of the transversal logical CNOT operator via data qubit permutation for GB codes constructed using a self-inverse polynomial.
These results highlight the potential of GB codes as efficient, fault-tolerant quantum codes and motivate further exploration of new families with higher encoding rates.

\section*{Acknowledgement}

This work has been supported by the Spanish Ministry of Economy and Competitiveness through the MADDIE project (Grant No. PID2022-137099NB-C44) 
and GEMMA project ``Generative AI and quantum approaches for advanced sensing and communications'' (Grant No. PID2025-169068NB-C33); and the Ministry of Economic Affairs and Digital Transformation of the Spanish Government through the QUANTUM ENIA project call - Quantum Spain project, and by the European Union through the Recovery, Transformation and Resilience Plan - NextGenerationEU within the framework of the “Digital Spain 2026 Agenda”. 
We would like to thank the members of the Quantum Information Group at Tecnun-University of Navarra for their support and valuable discussions.

\bibliographystyle{unsrt}
\bibliography{Gbfamilies}
\appendices

\section{GB codes and additive cyclic codes}\label{Ap:S2}
In this section, we highlight a natural connection between additive cyclic codes and GB codes. 
Given the extensive literature on additive cyclic codes, this connection enables us to take advantage of existing results to construct new families of GB codes, and provide minimum distance bounds for them.

Let $n$ be a positive integer and $\F_4=\{0,1,w,w^2\}$ be the quaternary field, where $w+1=w^2$. 
An $\F_2$-subspace $C$ of $\F_4^n$ is called an {\em additive code} over $\F_4$. 
The {symplectic inner product} of two vectors stated can be naturally redefined over $\F_4^n$ as
\begin{equation}\label{E:symplectic2}
\begin{split}
&(a_i+wb_i)_{i=1}^{n}  \ast (c_i+wd_i)_{i=1}^{n}\\&=(a_i)_{i=1}^{n} \cdot (d_i)_{i=1}^{n} + (b_i)_{i=1}^{n} \cdot (c_i)_{i=1}^{n}\\&=\sum_{i=1}^{n}a_id_i+b_ic_i 
\end{split}
\end{equation}
An additive code consisting of all the vectors that are symplectic orthogonal to an additive code $C$ will be denoted by $C^{\bot_s}$ and will be called its {\em symplectic dual}.

Moreover, if the code $C$ is closed under cyclic shifts of codewords, then it is called and an {\em additive cyclic code} over $\F_4$. 
Additive cyclic codes over $\F_4$ have been studied extensively in the literature, and many examples and families of good quantum codes are based on them \cite{kovalev2011design,bierbrauer2000quantum,Rezaadditive, dastbasteh2025additive}.  

Recall also that there is an $\F_2$-isomorphism $\phi:\F_4 \rightarrow\F_2^2$ that maps $\phi(0)=(0,0)$, $\phi(1)=(1,0)$, $\phi(w)=(0,1)$, and $\phi(w^2)=(1,1)$. 
This isomorphism can be extended to an $\F_2$-isomorphism $\psi:\F_4^n \rightarrow \F_2^{2n}$ defined by
\begin{equation}\label{E:2nrep}
\psi \big( (v_0,v_1,\ldots,v_{n-1}) \big)=(u_1,u_2),
\end{equation}
where we have $u_1=(u_{10},u_{11},\ldots,u_{1n-1})$ and  $u_2=(u_{20},u_{21},\ldots,u_{2n-1}) \in \F_2^{2n}$ with $\phi(v_i)=(u_{1i},u_{2i})$ for each $0\le i \le n-1$. 

Let $a(x)$ and $b(x)$ be two polynomials of $\F_2[x]/\langle x^n-1 \rangle$. Then 
the $\F_2$-linear span of 
$$\{x^i(a(x)+w b(x)): 0\le i \le n-1\}$$ 
forms a so called {\em one generator additive cyclic code} $C$ of length $n$ over $\F_4$. 
We represent such a code $C$ with $\langle a(x)+w b(x) \rangle$.
Applying the map $\psi$ to the code $C$ gives a binary linear code of length $2n$ with the generator matrix $[G_{a(x)}|G_{b(x)}]$. 
Using this transformation, one can
turn an additive cyclic code into a GB code of $(a(x),b(x))$. 
In other words, the corresponding GB code has the following properties: 
\begin{itemize}
\item the $X$ stabilizers are in correspondence to $\psi(C_1)$, where $C_1=\langle a(x)+w b(x) \rangle$. 
\item the $Z$ stabilizers are in correspondence to $\psi(C_2)$, where $C_2=\langle b(x^{-1})+wa(x^{-1}) \rangle$. 
\item $C_1$ and $C_2$ (hence $\psi(C_1)$ and $\psi(C_2)$) are equivalent. 
Therefore, their duals, and consequently the corresponding normalizer groups are equivalent. In short, $\psi(C_1) \sim \psi(C_2)$ and $\psi(C_1)^\bot \sim \psi(C_2)^\bot$.  
\end{itemize}
It should be noted, as shown in \cite[Theorem 3.6]{Rezaadditive}, that the maximum number of generators of an additive cyclic code is two. Therefore, the above transformation can also be used to extend beyond single-generator additive cyclic codes.

Let $C=\langle a(x)+w b(x) \rangle$.
We fix the following notation for additive cyclic codes:
\begin{itemize}
    \item The {\em reciprocal code} of $C$ is defined by $C^R=\langle a(x^{-1})+w b(x^{-1}) \rangle$. 
    \item The {\em conjugate code} of $C$ is defined by $\overline C=\langle b(x)+w a(x) \rangle$. 
\end{itemize}
One can connect the symplectic inner product and the Euclidean inner product of additive cyclic codes and their corresponding binary codes, respectively, under the map $\psi$ using the following result.

\begin{proposition}\label{P:symplectic-Euclidean}
Let $C=\langle a(x)+wb(x) \rangle$ be an additive cyclic code of length $n$ over $\F_4$. Then   
\[\psi((C)^{\bot_s})=\psi(\overline C)^\bot. \]
\end{proposition}

\begin{proof}
First recall that a polynomial $a(x)$ is in the dual of cyclic code generated by $b(x)$ iff $a(x)b(x^{-1})\equiv 0 \pmod{x^n-1}$. Second, \cite[Remark 4.1]{Rezaadditive} implies that $g_1(x)+wg_2(x)$ is in the symplectic dual of additive cyclic code generate by $g_1'(x)+wg_2'(x)$ iff 
$g_1(x)g_2'(x^{-1})+g_2(x)g_1'(x^{-1})\equiv0 \pmod{x^n-1}$.
Now we proceed with the proof. 

Let $v=m(x)+wn(x) \in C$. Note that \[\psi(\overline v)=\psi(n(x)+wm(x))=(n,m)\in \F_2^n \times \F_2^n,\] 
where $n$ and $m$ are the vector representations of $n(x)$ and $m(x)$, respectively.  
Hence we have $(p(x),q(x))\in \psi(\overline C)^\bot$ if and only if $p(x)\cdot n(x^{-1}) +q(x) \cdot m(x^{-1})\equiv 0 \pmod 2$ for each $v=m(x)+wn(x) \in C$. 
Using the second remark above, we have
\[
\begin{split}
0&\equiv p(x)\cdot n(x^{-1}) +q(x) \cdot m(x^{-1})\equiv (p(x)+wq(x))\ast \\&(m(x)+wn(x)) \pmod{x^n-1}.
\end{split}
\]
This is true if and only if $(p(x)+wq(x)) \in (C)^{\bot_s}$ which happens if and only if $(p,q)\in \psi( (C)^{\bot_s})$. 
Thus 
\[\psi((C)^{\bot_s})=\psi(\overline C)^\bot. \]
$\hfill \square$ 
\end{proof} 

The above connection between additive cyclic codes and binary linear codes enables us to systematically build binary quantum CSS codes from additive cyclic codes over $\F_4$. The following theorem summarizes this result.

\begin{theorem}\label{T:cyclicToGB}
Let $C=\langle a(x)+wb(x) \rangle$  be an additive code of length $n$ over $\F_4$ with $\dim(C)=k$ (dimension over $\F_2$).  
Then there exists a binary quantum (which is a GB) code with parameters $[\![2n,2(n-k),d]\!]_2$, where 
$$d \ge \min \{d( (C^R)^{\bot_s} \setminus C)\} \ge d((C^R)^{\bot_s}).$$ 
In particular, if $C$ is palindromic, i.e., $C=C^R$, then 
$$d \ge \min \{d(C^{\bot_s}\setminus C)\} \ge d(C^{\bot_s}).$$ 
\end{theorem}

\begin{proof}
Let $g(x)=a(x)+wb(x)$.
Using the argument given in \cite[Remark 4.1]{Rezaadditive}, we have $C$ and $C^R$ are symplectic orthogonal as
\[
\begin{split}
g(x) &\ast g(x)^R=\big(a(x)+wb(x) \big)\ast \big(a(x^{-1})+wb(x^{-1}) \big) \\& \equiv a(x)b(x)+a(x)b(x)\equiv 0 \pmod{x^n-1}.
\end{split}
\]
Thus $C \subseteq (C^R)^{\bot_s}.$ 
Now applying the map $\psi$ to both sides gives 
$$\psi(C) \subseteq \psi\big ((C^R)^{\bot_s} \big).$$
The quantum CSS code of $\psi(C) \subseteq \psi\big ((C^R)^{\bot_s} \big)$ has dimension 
\[\dim \big( \psi \big ((C^R)^{\bot_s} \big) \big)-\dim \big( \psi(C) \big)=(2n-k)-k=2(n-k).\]
Moreover, we have 
\[d\big( \psi ((C^R)^{\bot_s}) \setminus  \psi(C) \big)=d\big( \psi\big(({C}^R)^{\bot_s}) \setminus  C \big)\big)
\]
Hence such a CSS code has minimum distance 
$$d \ge \min \{d( (C^R)^{\bot_s}\setminus C)\} \ge d((C^R)^{\bot_s}).$$
The left inequality is due to the fact that $\psi$ sends a vector of weight $t$ to a vector of weight ``at least'' $t$. 

In the case of $C=C^R$, the result follows immediately.$\hfill \square$ \end{proof} 

It should be noted that the above distance inequality can be strict. For instance, when all minimum weight codewords of $ (C^R)^{\bot_s}\setminus C$ have at least one coordinate position containing $w^2$. In this case, we have 
\[
d > d\big( (C^R)^{\bot_s}\setminus C \big).\]
On the other hand, if there exists a minimum weight codeword of $ C_2^{\bot_s}\setminus C_1$ that only has coordinate values belonging to $\{0,1,w\}$, then the above minimum distance bound becomes equality, i.e., 
\[
d= d\big( (C^R)^{\bot_s} \setminus  C \big).
\]  
Furthermore, one can generalize the result of Theorem \ref{T:cyclicToGB} to cover more general additive codes (with more than one generator). 
In particular, a similar argument as the one in the previous proof implies that replacing $C$ with any symplectic self-orthogonal additive code (even with more than one generator) leads to the same result.   

Through the next example, we give an instance of the former case when the minimum distance improvement happens. 

\begin{example}\label{E:26-2-6}
Let $n=13$. Then applying the result of Theorem \ref{T:cyclicToGB} to the symplectic self-orthogonal code $C$ obtained from Theorem 6.4 of \cite{dastbasteh2025additive} 
(which is in correspondence to a non-CSS code $[\![13,1,5]\!]_2$) implies a CSS code with parameters $[\![26,2,6]\!]_2$,
which has the same rate but a larger minimum distance. 
Here the qubit connectivity is larger than four.   
\end{example}

Next we discuss a natural minimum distance upper bound for GB codes.

\begin{lemma}\label{L:upper}
Let $a(x)$ and $b(x)=m(x)a(x)$ be two polynomials of degree less than $n$. If the quantum GB code of length $2n$ corresponding to polynomials of $a(x)$ and $b(x)$ has dimension larger than zero, then  it has minimum distance of at most $\wt (m(x))+1$.
\end{lemma}
\begin{proof}
The corresponding GB code
has the parity check matrix
\[H=\begin{bmatrix}H_x & 0\\
0 & H_z
\end{bmatrix},\]
where 
$H_x=[G_{a(x)},G_{b(x)}]$ and $H_z=[G_{b(x^{-1})},G_{a(x^{-1})}]$.
This quantum code is the CSS code of $C_2 \subseteq C_1$, where 
$H_x$ and $H_z$ are parity check matrices of $C_2^\bot$ and $C_1$, respectively. 
Let $w=(u,v)$, where $u=(1,0,\ldots,0)$ and $v$ is the vector representation of $m(x)$. 
Then $w\in C_1$ because 
\[
\begin{split}
&H_zw^T=b(x^{-1})u^T+a(x^{-1})v^T=\\&b(x^{-1})+a(x^{-1})m(x^{-1})=b(x^{-1})+b(x^{-1})=0.
\end{split}
\]
Since each $(u',v') \in C_2$  satisfies $\wt(u')>1$, as otherwise the GB code has dimension zero, we conclude that $w \not\in C_2$ and it has weight $\wt(m(x))+1$. Thus $d(C_1 \setminus C_2) \le \wt(m(x))+1$.    
$\hfill \square$ 
\end{proof} 

The above lemma establishes a relationship between the minimum distance of a GB code defined by the polynomials $a(x)$ and $b(x)$ and the structure of the polynomial $b(x)$. 
It also provides a straightforward way to derive an upper bound on the minimum distance.

We consider two binary codes as {\em permutation equivalent} if there exists a permutation of coordinates that maps one code to another. 
Some equivalence criteria for GB codes is discussed in Statement 3 of \cite{wang2022}. 
The following proposition gives a new condition for GB codes to have the same parameters.

\begin{proposition}\label{P:equivalence}
Let $a(x)$ and $b(x)$ be two binary polynomials of $\F_2[x]/\langle x^n-1 \rangle$. 
Then the GB codes corresponding to $a(x)$ and $b(x)$, and $a'(x)=x^ia(x)$ and $b'(x)=x^jb(x)$ both have the same parameters (length, dimension, and minimum distance) for each $0\le i,j\le n-1$. 
\end{proposition}
\begin{proof}
These GB codes are the CSS codes of $C_2 \subseteq C_1$ and $C_2' \subseteq C_1'$, where 
$C_2$ and $C_2'$ have the generator matrices
$H_x=[G_{a(x)},G_{b(x)}]$ and $H_x'=[G_{a'(x)},G_{b'(x)}]$,  
and $C_1$ and $C_1'$ have the check matrices $H_z=[G_{b(x^{-1})},G_{a(x^{-1})}]$ and $H_z'=[G_{b'(x^{-1})},G_{a'(x^{-1})}]$, respectively. 

Applying $i$ cyclic shifts to the first, and $j$ shifts to the second $n$ components of $C_2$ (that is a permutation action) sends the matrix $H_x$ to $H_x'$. So $C_2$ and $C_2'$ are equivalent, and hence have the same parameters and weight distribution. Similarly, applying $n-j$ shifts to the first and $n-i$ shifts to the second $n$ components of $C_1$ maps it to $C_1'$. Thus we have 
\begin{itemize}
\item  $C_1$ and $C_1'$ are permutation equivalent, 
\item  $C_2$ and $C_2'$ are permutation equivalent,
\item  $C_2 \subseteq C_1$ and $C_2' \subseteq C_1'$.
\end{itemize}
The first two conditions imply that the corresponding GB codes have the same length and dimension.
Moreover, these three conditions imply that $C_1 \setminus C_2$ and $C_1' \setminus C_2'$ (respectively $C_2^\bot \setminus C_1^\bot$ and $C_2'^\bot \setminus C_1'^\bot$) have the same minimum weight vectors. 
$\hfill \square$ 
\end{proof} 

This result reduces the search space for GB codes by a factor of $n^2$, facilitating systematic construction with desired properties. 
Equivalence relations of cyclic codes, including monomial and  multipliers \cite{dastbasteh2023equivalence,dastbasteh2025equivalence,sendrier2002finding,grassl2013leveraging}, can efficiently identify additional equivalent GB codes and implement logical operations \cite{grassl2013leveraging,sayginel2024fault}, which we exploit to realize a logical CNOT gate in later sections.

\section{GB distance bounds and families}\label{Ap:s3}

\begin{example}[True minimum distance calculation]\label{Exa:2}
Let $n=27$, $f(x)=1+x+x^2$, and $p(x)=1+x^3+x^6\cdots+x^{18}$. Next we apply Theorem IV.1 and show that the GB code corresponding to the polynomials $f(x)$ and $f(x)P(x)$ has parameters $[\![54,4,6]\!]_2$. First, since $\deg(f(x))=2$ the claim about dimension follows immediately. 

Note that the cyclic code generate by $h(x)=\frac{x^n-1}{f(x)}$ has distance $\frac{2n}{3}=18$. Also one can verify that 
\[
h(x)=p(x)(x^6 + x^7)+(1+x+x^3+x^4).
\]
and 
\[
p(x)^{-1}=x^3+x^6+x^{12}+x^{18}+x^{24}.
\]
Now applying the minimum distance bound of Theorem IV.1 implies that 
\[
\min\{\frac{12\times 18}{42}\approx 5.14 , \frac{18}{4}\}\le d_{GB}\le \min\{8,6\} .
\]
This gives the bound $5\le d_{GB}\le 6$. Next, we use (IV.6) to improve the lower bound $\frac{18}{4}$ given above. 
Note that as we mentioned earlier all the $(u(x),v(x)) \in C_1$ that are equivalent to an $X$-normalizer satisfy 
\[\textbf{(a)}\ \begin{split}
 u(x)p(x)+v(x) \equiv 0 \pmod{x^n-1} 
 \end{split}
 \]
or 
\[\textbf{(b)}\
\begin{split}
& u(x)p(x)+v(x) \not\equiv 0 \pmod{x^n-1} \\ &
\text{and} \ h(x)\mid u(x)p(x)+v(x).
\end{split}
\]
Moreover, all satisfying (b) have the property that $5<\wt(u(x),v(x))$ by the above lower bound. All $(u(x),v(x))$ satisfying (a)  have the property that $v(x)=u(x)p(x)$ and 
\[18 \le 2\wt(u(x)p(x))+4\wt(u(x)).\]
If $\wt(u(x)) \le 3$, then $\wt(v(x))\ge 3$ and thus $6\ge \wt((u(x),v(x)) )$.
If $\wt(u(x))=4$, then $\wt(v(x))\ge 2$ and again $6\ge \wt((u(x),v(x)))$. 
Finally, If $\wt(u(x))=5$, then $\wt(v(x))\ge 1$ (because $u(x)=0$ iff $u(x)p(x)=0$) and again $6\ge \wt((u(x),v(x)))$.

Therefore for all vectors satisfying either (a) or (b) we have $6\ge \wt((u(x),v(x)))$ which implies $6 \le d_{GB}$. Combining it with the previous bound, we obtain $d_{GB}=6$ and this is a $[\![54,4,6]\!]_2$ GB code.

It should also be noted that $((1+x)f(x),(1+x)f(x)p(x))=(1+x^3,1+x^{21})$ is a weight four stabilizer and hence this GB code is a degenerate code. 
\end{example}
Here again, the above quantum code is degenerate (to $4$), yet our distance bound captures the true minimum distance, which is $d=6$.

\begin{proof}[Proof of Lemma IV.1]
First, we compute the product of $P_d(x)P_d(x^d)$ in $\F_2[x]$. Recall that $2n=d^2+1$. 
Then we have 
\begin{equation}\label{E:Pd}
\begin{split}  
P_d(x)P_d(x^d)&= \frac{1+x^d}{1+x}\frac{1+x^{d^2}}{1+x^d}=\frac{1+x^{d^2}}{1+x}\\&=1+x+\cdots+x^{2n-2}.
\end{split}  
\end{equation}
Thus 
\[
\begin{split}  
P_d(x)&(xP_d(x^d))\equiv(x+x^2+\cdots+x^{n-1})\\&+(1+x+x^2+\cdots+x^{n-1}) \equiv 1 \pmod{x^n-1}
\end{split}  
\]
and 
\[
\begin{split}  
P_d(x)&(xP_d(x^d))\equiv(x+x^2+\cdots+x^{n-1})\\&+x^n(1+x+x^2+\cdots+x^{n-1}) \\&\equiv 1 \pmod{P_n(x)}.
\end{split}  
\]
The last part of (1) follows by counting the number of terms in $xP_d(x^d)$, which is $d$.

The second part follows from a straightforward calculation.
$\hfill \square$ \end{proof} 

\begin{proof}[Proof of Theorem IV.3]
As we showed in the proof of Theorem IV.2, for $u(x)=1+x^d+x^{2d}+\cdots+x^{(\frac{d-1}{2}-1)d}$, the vector $(u(x),v(x)=u(x)P_d(x)+P_n(x))$, which has weight $d$, is a logical operator. Hence it can be chosen as $XI$ logical operator. Moreover, (IV.1) implies 
\[
\begin{split}
P_n(x)&=P_d(x)u(x)+v(x)=\\&P_d(x)(1+x^d+x^{2d}+\cdots+x^{d(\frac{d-1}{2}-1)})\\&+(x^{d\frac{d-1}{2}}+x^{d\frac{d-1}{2}+1}+\cdots+x^{n-1}).
\end{split}
\]
Substituting $x$ with $x^d$ and multiplying both sides by $x$ (and computing modulo $x^n-1$) implies that
\[
\begin{split}
P_n(x)&=xP_d(x^d)u(x^d)+xv(x^d)\\&=P_d(x)^{-1}u(x^d)+xv(x^d).
\end{split}
\]
Hence $(u(x^d)P_d(x)^{-1}+P_n(x),u(x^d))=(xv(x^d),u(x^d))$ is also a logical operator.
This operation has different parity in compare to the logical operation corresponding to $XI$ (their sum is not a stabilizer). 
So we choose it to be $XX$. 
Finally, note that $(u(x),v(x))+(xv(x^d),u(x^d))$ has an odd weight vector in the first and the second component. So the minimum weight of a logical $IX$ operator is at least $d+1$. 
Indeed, one can fix the operation $IX$ to be $(1,P_d(x))$ or $(P_d(x)^{-1},1)$.  
$\hfill \square$ 
\end{proof}

\begin{proof}[Proof of Theorem V.2]
Consider the GB code with the associated polynomials $f(x)=1+x$ and $f(x)P_{d+1}(x)$. 
The proof follows in a similar fashion as the proof of Theorem IV.2, except the argument about the minimum distance. Therefore, we only prove the claim about the minimum distance.

Let $d_{GB}$ be the minimum distance of such code. Using the above lemma, we have $\wt(P_{d+1}(x))=d+1$, $\wt(P_{d+1}(x)^{-1})=d-1$, and 
  \[
\begin{split}
P_n&=P_{d+1}(x)^{-1}s(x)+r(x),
\end{split}
\]
where $\wt(s(x))=\frac{d}{2}$ and $\wt(r(x))=\frac{d}{2}$.

Applying the minimum distance bounds given in Theorem III.1 implies that 
\[
\min \{\frac{2dn}{d(d+1)}, \frac{n}{\frac{d}{2}}\}\le d_{GB} \le \min \{d+2,d,d\}
\]
which implies that
\[
\min \{\frac{d^2}{(d+1)}, d\}\le d_{GB} \le \min \{d\}.
\]
Therefore, $d-1 < d_{GB}\le d$, which forces the minimum distance to be $d_{GB}=d$. The degeneracy part follows immediately as the stabilizer generators all have weight four.
$\hfill \square$ \end{proof} 

\subsection{Hook error and effective weight}\label{Ap:s4}

A particularly dangerous class of syndrome measurement circuit faults
arises when a single check qubit error propagates through syndrome extraction circuit, i.e., entangling operations such as CNOT gates, causing multiple data qubit errors, see Figure \ref{F:hook}. 
These errors are particularly problematic because they can lead to a logical error at lower cost if the two induced data qubit errors align with a minimum weight logical operator (so called “bad hooks”), making them more difficult to detect and correct. 
Therefore, avoiding bad hook errors can enhance the overall performance of an error-correcting code.

Also, it should be emphasized that only check qubit errors propagating to
data qubits can reduce the effective distance, as data qubit errors never spread to neighboring data qubits. 

\begin{figure}[h]
    \centering
\[
\Qcircuit @C=1em @R=1em {
    & \lstick{|0\rangle} & \gate{H} & \ctrl{1} & \ctrl{2} & \gate{\textcolor{red}{X}} & \ctrl{3} & \ctrl{4} & \gate{H} & \qw & \qw & \meter  \\
    & \lstick{q_0} & \qw & \targ & \qw & \qw & \qw & \qw & \qw & \qw & \qw& \qw\\
    & \lstick{q_1} & \qw & \qw & \targ & \qw & \qw & \qw & \qw & \qw & \qw & \qw\\
    & \lstick{q_2} & \qw & \qw & \qw & \qw & \targ & \qw & \qw & \qw & \gate{\textcolor{red}{X}} & \qw\\
    & \lstick{q_3} & \qw & \qw & \qw & \qw & \qw & \targ & \qw & \qw & \gate{\textcolor{red}{X}}& \qw 
}
\]
\caption{An example of a hook error, where one error on the check qubit (top qubit) spreads to two data qubits ($q_2$ and $q_3$).}
\label{F:hook}
\end{figure}

Next, we give more details about RL syndrome extraction scheme for the $[\![d^2+1,2,d]\!]_2$ family and show that, in this setting, hook errors do not reduce the code’s minimum distance.
Our proof is based on the minimum distance bounds developed in the previous sections.
Since applying a permutation to all $X$-stabilizers generates all $Z$-stabilizers, we restrict our discussion to the former resulting in a proof for both cases.

It is important to note that {\em a general syndrome extraction pattern for such GB family is not necessarily immune to hook errors}.
Indeed, at the end of this section we provide an example of a syndrome extraction pattern in which hook errors align with the structure of minimum-weight logical errors, thereby reducing the effective distance of the code.

First recall that there exists a basis of the $X$-stabilizer group which consists of $x^i(1+x,1+x^d)$, for each $0\le i \le n-2$. 
So, in order to extract the syndrome, one needs to only operate on such basis with the aid of $n-1$ check qubits. 
Moreover, recall that
the $i$-th check qubit acts as $x^i(1+x,1+x^d)$, namely it acts on physical qubits $i$, $i+1$ (corresponding to $x^i+x^{i+1}$ so called left qubits) and $x^i+x^{i+d}$ (corresponding to $x^i+x^{i+d}$ so called right qubits). Thus if we represent the first $n$ qubits by $q_0,q_1,\ldots, q_{n-1}$, and the second $n$ qubits by $q_0',q_1',\ldots, q_{n-1}'$, we conclude that the $i$-th check operates on $q_i$, $q_{i+1}$, $q'_i$, and $q'_{i+d}$.

The syndrome extraction order that we choose is to first operate as  $(0,x^i(1+x^d))$ and then $(x^i(1+x),0)$. 
We call such order {\em Right-Left (RL)} and it is shown in Figure \ref{F:RL}. 
\begin{figure}[h]
    \centering

\[
\Qcircuit @C=1em @R=1em {
    & \lstick{|0\rangle} & \gate{H} & \ctrl{1} & \ctrl{2} & \gate{\textcolor{red}{X}} & \ctrl{3} & \ctrl{4} & \gate{H} & \qw & \qw & \meter  \\
    & \lstick{q_i'} & \qw & \targ & \qw & \qw & \qw & \qw & \qw & \qw & \qw& \qw\\
    & \lstick{q_{i+d}'} & \qw & \qw & \targ & \qw & \qw & \qw & \qw & \qw & \qw & \qw\\
    & \lstick{q_{i}} & \qw & \qw & \qw & \qw & \targ & \qw & \qw & \qw & \gate{\textcolor{red}{X}} & \qw\\
    & \lstick{q_{i+1}} & \qw & \qw & \qw & \qw & \qw & \targ & \qw & \qw & \gate{\textcolor{red}{X}}& \qw 
}
\]
\caption{RL extraction pattern for the $i$-th check qubit corresponding to $x^i(1+x,1+x^d)$ X-stabilizer, and propagation of an error.}
\label{F:RL}
\end{figure}
This pattern implies that all the hook errors resulted from a single error in the check qubit will be in the form  of $(x^i(1+x),0)$ for some $0\le i \le n-2$. 
All other type single errors on a check qubits can be translated to an equivalent single data qubit error. 
So we discard them in our conversation as they do not induce error propagation. 
An arbitrary combination of hook errors has the form 
$(e(x)(1+x),0)$ for some $e(x) \in \F_2[x]/\langle x^n-1 \rangle$. 
The number of errors in the check qubits, in order to obtain such combination, is $\wt(e(x))$.

Next we show that such errors would not reduce the minimum distance of the GB codes under discussion. 
Recall that effective minimum distance of a code is less than or ideally equal to its theoretical minimum distance.

\begin{proof}[Proof of Theorem IV.5]
Recall that $n=\frac{d^2+1}{2}$.
Let $e=(u(x),v(x))+(j(x)(1+x),0)$ be a non-trivial logical error implied by $\wt(j(x))$ check qubits errors, and the arbitrary $(u(x),v(x))$ error on the data qubits.
Hence the number of errors required to obtain $e$ is $\wt(j(x))+\wt(u(x))+\wt(v(x))$, and we show that it is at least $d$. 

Recall from (IV.1) that such error satisfies 
\[
\begin{split}
(1+x)&\big(P_d(x)(j(x)(1+x)+u(x))+v(x)\big)\\& \equiv 0 \pmod{x^n-1}.
\end{split}
\]
Hence two cases can happen:
\[
\begin{split}
\big(P_d(x)(j(x)&(1+x)+u(x))+v(x)\big) \\&\equiv 0 \pmod{x^n-1}
\end{split}
\]
or
\[
\begin{split}
\big(P_d(x)(j(x)&(1+x)+u(x))+v(x)\big) \\&\equiv P_n(x) \pmod{x^n-1}.
\end{split}
\]
First we consider the former case. Recall that, as we showed in Lemma  IV.1, we have
\[
\begin{split}
P_n&=P_d(x)(1+x^d+x^{2d}+\cdots+x^{d(\frac{d-1}{2}-1)})\\&+(x^{d\frac{d-1}{2}}+x^{d\frac{d-1}{2}+1}+\cdots+x^{n-1}).
\end{split}
\]
Multiplying all sides by $j(x)(1+x)+u(x)$ implies that
\[
\begin{split}
P_n&=v(x)(1+x^d+x^{2d}+\cdots+x^{d(\frac{d-1}{2}-1)})\\&+(j(x)(1+x)+u(x))(x^{d\frac{d-1}{2}}+x^{d\frac{d-1}{2}+1}\\&+\cdots+x^{n-1}).
\end{split}
\]
Next we bound the weight of right hand side using the left one. In particular, using the fact that 
$(1+x)(x^{d\frac{d-1}{2}}+x^{d\frac{d-1}{2}+1}+\cdots+x^{n-1})$ has weight two, we have 
\[
\begin{split}
n &\le \frac{d-1}{2} \wt(v(x))+\frac{d+1}{2} \wt(u(x)) +2\wt(j(x))\\& \le \frac{d+1}{2}(\wt(j(x))+\wt(u(x))+\wt(v(x))).
\end{split}
\]
Therefore,
\[
d \le \wt(j(x))+\wt(u(x))+\wt(v(x))
\]
which implies that the number of faults should be at least $d$ in order to achieve such logical operator. 

Next we consider the second possibility above. 
we have 
\begin{enumerate}
    \item $P_d(x)(j(x)(1+x)+u(x))+v(x)\equiv P_n(x) \pmod{x^n-1}$
    \item and multiplying by $P_d^{-1}(x)$ gives $j(x)(1+x)+u(x)+P_d(x)^{-1}v(x)\equiv P_n(x) \pmod{x^n-1}$.
\end{enumerate}
Again calculating the weights of the sides implies 
\begin{itemize}
    \item from (1): $n \le 2\wt(j(x))+d\wt(u(x))+\wt(v(x))$
    \item and from (2): $n \le 2\wt(j(x))+\wt(u(x))+d\wt(v(x))$.
\end{itemize}
Adding them gives $2n \le 4\wt(j(x))+(d+1)\wt(u(x))+(d+1)\wt(v(x))$ which implies
 $2n \le (d+1)\big( \wt(j(x))+\wt(u(x))+\wt(v(x))$ and thus $d\le \wt(j(x))+\wt(u(x))+\wt(v(x))\big)$.
Thus this GB family has effective distance $d$. 
$\hfill \square$ \end{proof} 

It should be noted that replacing the RL extraction with LR implies the same result as above. 
It is mainly because that each stabilizer in the form of $x^i(1+x,1+x^d)$ can be decomposed as $(x^i(1+x),0)+(0,x^i(1+x^d))$. 
Therefore, RL extraction preserves the minimum distance if and only if LR also preserves the minimum distance. 

Our computations show that 
replacing the RL extraction with other 
syndrome extraction patterns
can reduce the effective minimum distance to approximately half of the theoretical minimum distance. 
One example of such patterns is to first operate as $(x^i,x^i)$ and then as $(x^{i+1},x^{i+d})$, and one can obtain the corresponding circuit by swapping the order of two middle CNOTs in Figure \ref{F:RL}. 
In this case, our computations, employed in Magma Computer Algebra System \cite{magma}, for $3 \le d  \le 11$ show that the effective minimum distance of such GB codes is $\lceil \frac{d}{2} \rceil$.

\section{GB codes with nearly self-inverse $p(x)$ and connectivity six family} \label{Ap:s4}

\subsection{Existence of $[[2(d+1)(2d+1),2,2(d+1)]]$ family}
In this section, we provide a detailed justification for the existence of this family. 

Let $d$ be a positive odd integer and $n=(d+1)(2d+1)$. We fix $c=\frac{n}{2}$, $p(x)=P_d(x)+x^{c}P_{d+1}(x)$, and $a(x)=1+x$. 
Then $b(x)=(1+x)p(x)=1+x^d+x^c+x^{1+c+d}$. 
Hence the corresponding GB code has the qubit connectivity of six.
A simple calculation shows that $p(x)^{-1}=x^{-2d}p(x)$ modulo $x^n-1$.

Moreover, consider
\[
s(x) = \sum_{j=0}^{\frac{d+1}{2}} x^{jd} + \sum_{m=1}^{\frac{d-1}{2}} x^{\frac{d(d+1)}{2} + m(d+1)} \in \mathbb{F}_2[x]/(x^n-1),\] which has $\wt(s(x)) = d+1$. 
Multiplying $s(x)$ by $p(x) = P_d(x) + x^c P_{d+1}(x)$ produces $s(x)p(x) = \sum_{i \in I} x^i$, where the support $I \subset \{0, \dots, n-1\}$ misses exactly $d+1$ indices due to the $\frac{d-1}{2}$ inter-block gaps in $s(x)P_d(x)$, the term $x^{c-1}$, and $\frac{d+1}{2}$ single-term overlaps in $x^c s(x)P_{d+1}(x)$. 
Thus, 
\[
\begin{split}
r(x) &= P_n(x) + s(x)p(x) = \sum_{m=0}^{\frac{d-3}{2}} x^{\left(\frac{d+3+2m}{2}\right)d + m} \\&+ x^{d^2 + \frac{3d-1}{2}} + \sum_{j=1}^{\frac{d+1}{2}} x^{c + jd}
\end{split}
\]
which has $\wt(r(x)) = d+1$.
Using the fact that $\wt(p(x))=2d+1$, the result of Theorem VI.1 implies that 
\[
\begin{split}
&2d+1= \min\{\frac{2n}{2(d+1)},\frac{2n}{2(d+1)}\} \le d_{GB} \\&\le \min\{2(d+1),2(d+1)\}=2d+2.    
\end{split}
\]
Finally, all the logical operators are in the form of $(u(x),p(x)u(x))$ and since $n$ is even, we have $\wt(u(x))\equiv \wt (p(x)u(x)) \pmod 2$ which implies that $d_{GB}$ is even. 
Thus there exists a family of GB codes with connectivity six and parameters $[[2(d+1)(2d+1),2,2(d+1)]]$.

\subsection{Self-inverse polynomials}

Let $n$ and $i$ be positive integers such that $n$ is odd and $\gcd(i, n) = 1$. 
Let $R = \mathbb{F}_2[x]/\langle x^n - 1 \rangle$. By the Chinese Remainder Theorem, $R$ decomposes uniquely into a direct sum of extension fields:
\[
R \cong \bigoplus_{j=1}^r K_j, \quad \text{where } K_j = \frac{\mathbb{F}_2[x]}{\langle f_j(x) \rangle} \cong \mathbb{F}_{2^{d_j}}
\]
where $x^n - 1 = \prod_{j=1}^r f_j(x)$ is the factorization into pairwise coprime irreducible polynomials over $\mathbb{F}_2$, and $d_j = \deg(f_j(x))$.

Using the result of Theorem 4.3.8 of \cite{Huffman}, one can find the primitive idempotents $\{e_1(x), \dots, e_r(x)\} \subset R$ satisfying:
\begin{enumerate}
    \item $e_j(x) e_k(x) \equiv 0 \pmod{x^n - 1}$ for $j \neq k$,
    \item $e_j(x)^2 \equiv e_j(x) \pmod{x^n - 1}$,
    \item $\sum_{j=1}^r e_j(x) \equiv 1 \pmod{x^n - 1}$.
\end{enumerate}
Hence, every polynomial $p(x) \in R$ admits a unique representation in terms of these idempotents:
\[
p(x) = \sum_{j=1}^r u_j e_j(x) \pmod{x^n - 1},\]
where $u_j = p(x) \bmod f_j(x) \in K_j$.

The mapping $\sigma_i: R \to R$ defined by $\sigma_i(a(x)) = a(x^i) \bmod (x^n - 1)$ is a ring automorphism. 
This map:
\begin{enumerate}
    \item Permutes the primitive idempotents as $\sigma_i(e_k(x)) \equiv e_{\pi(k)}(x) \pmod{x^n - 1}$ for a permutation $\pi$ acting on $\{1,2,\ldots,r\}$.
    \item Induces field isomorphisms $\tau_k: K_k \xrightarrow{\cong} K_{\pi(k)}$ given by $\tau_k(a(x)) = a(x^i) \bmod f_{\pi(k)}$.
\end{enumerate}
Applying $\sigma_i$ to $p(x)$ yields
\[\begin{split}
\sigma_i(p(x)) &= \sum_{k=1}^r \tau_k(u_k) e_{\pi(k)}(x) \\&= \sum_{j=1}^r \tau_{\pi^{-1}(j)}(u_{\pi^{-1}(j)}) e_j(x).
\end{split}
\]
Now, the target equation $p(x)p(x^i) \equiv 1 \pmod{x^n - 1}$ holds if and only if, for every component $j \in \{1, \dots, r\}$:
\[
u_j \cdot \tau_{\pi^{-1}(j)}(u_{\pi^{-1}(j)}) = 1_{K_j} \iff u_{\pi(k)} = \left( \tau_k(u_k) \right)^{-1}.
\]

Let $\mathcal{O} = (j_0, j_1, \dots, j_{L-1})$ be an orbit of the permutation $\pi$ of length $L$, where $j_{k+1} = \pi(j_k)$ for $0 \le k < L-1$.
For each orbit $\mathcal{O}$, choosing a single initial component $u_{j_0} \in K_{j_0}^*$ uniquely and deterministically forces all remaining components in the orbit.
Because the orbit is a closed loop ($j_L = j_0$), propagating $u_{j_0}$ through $L$ steps yields $u_{j_L} = u_{j_0}$. 

Let $\mathcal{T}_L = \tau_{j_{L-1}} \circ \dots \circ \tau_{j_0} \in \text{Aut}(K_{j_0})$, which acts as a Frobenius power $\mathcal{T}_L(y) = y^{2^S}$ for some $0 \le S < d_{j_0}$.

The initial choice $u_{j_0} \in K_{j_0}^*$ is valid if and only if it satisfies the closure constraint:
\[
\begin{cases}
\mathcal{T}_L(u_{j_0}) = u_{j_0} \iff u_{j_0}^{2^S - 1} = 1 & \text{if } L \text{ is even}, \\
u_{j_0} \cdot \mathcal{T}_L(u_{j_0}) = 1 \iff u_{j_0}^{2^S + 1} = 1 & \text{if } L \text{ is odd}.
\end{cases}
\]

Now we have all the necessary tools for characterizing all self-inverse polynomials.
\begin{theorem}[Self invertibility condition]
Let $n$ be an odd integer.
A polynomial $p(x) \in \mathbb{F}_2[x]/\langle x^n - 1 \rangle$ satisfies $p(x)p(x^i) \equiv 1 \pmod{x^n - 1}$ if and only if it is constructed by:
\begin{enumerate}
    \item Decomposing the primitive idempotents into orbits under $x \mapsto x^i$.
    \item Selecting exactly one initial element $u_{j_0} \in K_{j_0}^*$ satisfying the closure condition for each orbit $\mathcal{O}$ (given above this theorem).
    \item Setting $p(x) = \sum_{j=1}^r u_j e_j(x) \pmod{x^n - 1}$.
\end{enumerate}
\end{theorem}

When  $n$ is even, one can form many self-inverse polynomials simply by choosing $p(x)$ of degree less than $\frac{n}{2}$ such that $p(0)\neq 0$. Then $q(x)=p(x)+x^{\frac{n}{2}}(1+p(x))$ is self-inverse as $q(x)^2\equiv 1 \pmod{x^n-1}$.

\begin{proof}[Proof of Theorem VI.4]
The first part follows from the above discussions. 

For the second part, 
let $p(x)p(x^r)\equiv 1 \pmod{x^n-1}$ for some positive integer $r$ such that $\gcd(r,n)=1$.
Moreover, one can partition the basis of the logical operators by choosing
\[
\begin{split}
\{
(x^ig(x),0),(x^{i}g(x),x^{ir}g(x^r))
\}    
\end{split}
\]
for each $0\le i \le \deg(f(x))-1$, where $g(x)=\frac{x^n-1}{f(x^{-1})}$ is the generator polynomial of dual of the cyclic code generated by $f(x)$. 
One can correspond these pairs into the basis states 
\[
\{\ket{s_{2i+1}},\ket{s_{2i+2}}\}
\]
where $s_{2i+1}$ and $s_{2i+2}$ are elements of $\mathbb{F}_2^{2\deg(f(x))}$ with zeros everywhere except in the positions $2i+1$ and $2i+2$, respectively.
We define the permutation of the data qubits as 
\[
E_rS,
\]
where
\[
\begin{split}
S\big((a(x),b(x))\big)\rightarrow (b(x),a(x)),
\end{split}
\]
and 
\[
\begin{split}
E_r\big((a(x),b(x))\big)\rightarrow (a(x^r),b(x^r)).
\end{split}
\]
One can easily verify that, under this action, each stabilizer element is mapped to another stabilizer element. 
Moreover, if we restrict our attention to locations $2i+1$ and $2i+2$ of each logical state, this action translates to 
\[
\ket{10} \longleftrightarrow \ket{11},
\ket{01} \longleftrightarrow \ket{01}.
\]
This is because $E_rS$ operates as
\[
\begin{split}
&(x^ig(x),0) \leftrightarrow (0,x^{ir}g(x^r)) \ \text{and} \\& \ (x^ig(x),x^{ir}g(x^r)) \leftrightarrow (x^ig(x),x^{ir}g(x^r)).
\end{split}
\]

Finally, we show that $E_rS$ preserves the stabilizer group. 
We only give a proof for the $X$-stabilizers and the other case follows from similar steps.
Let $C_2$ be the
binary space spanned by the $X$-stabilizers, i.e., by the simultaneous cyclic
shifts of of $g_0=(f(x),f(x)p(x))$, namely
\[
g_i=\big(x^if(x),\,x^if(x)p(x)\big),\qquad 0\le i\le n-1.
\] 
The
permutation $E_rS$ maps $g_i$ to
\[
g_i'=\big(x^{ir}f(x^r)p(x^r),\,x^{ir}f(x^r)\big)=x^{ir}g_0' .
\]
Let $C_2'$ be the span of $g_0',\dots,g_{n-1}'$. 
Since $E_rS$ is linear
and the map $i\to (ir) \pmod n$ is a group isomorphism of $\mathbb{Z}/n\mathbb{Z}$, we deduce that $C_2'$ is linear and closed under cyclic shifts, and therefore closed under multiplication of both components by any polynomial.

Since $\langle f(x)\rangle=\langle f(x^r)\rangle$, there is a binary polynomial
$s(x)$ such that $s(x)f(x^r)\equiv f(x)\pmod{x^n-1}$. Multiplying both
components of $g_0'$ by $p(x)s(x)$ and using $p(x)p(x^r)\equiv 1$ gives
\[
p(x)s(x)\,g_0'
=\big(f(x),\,f(x)p(x)\big)=g_0 .
\]
Thus $C_2\subseteq C_2'$. A permutation of coordinates does not change the dimension
of a linear code, so $C_2'=C_2$.
Hence the stabilizers are preserved under this permutation.
$\hfill \square$
\end{proof}
\begin{remark}
Let $\beta$ be a primitive $n$-th root of unity and
$Z(f)=\{j\in\mathbb{Z}/n\mathbb{Z} : f(\beta^j)=0\}$, which is the so called defining set of the cyclic code generated by $f(x)$. 
The condition
$ \langle f(x)\rangle= \langle f(x^r) \rangle$ is equivalent to $Z(f)$ being closed under multiplication by $r$.
Such a symmetry can be checked easily using the unique defining set of cyclic codes.
For instance, it always holds for $f(x)=1+x$, and $f(x)=1+x+x^2$ (when $3\mid n$).
\end{remark}

Below, we list a few examples of GB codes that satisfy the conditions of the above proof, and realize transversal logical operations via a physical permutation of data qubits.

\begin{table}[ht]
\centering
\caption{GB codes realizing transversal logical $\text{CNOT}$ operations via physical data qubit permutations.}
\label{tab:gb_cnot_codes}
\begin{tabular}{ll}
\toprule
\textbf{Parameters $[\![ n, k, d ]\!]$} & \textbf{Logical Operation} \\
& \textbf{(Permutation Type)}\\
\midrule
$[\![ 22,2,6 ]\!], [\![ 30,2,7 ]\!],  [\![ 50,2,9 ]\!],[\![ 62,2,10 ]\!]$ & $\overline{\text{CNOT}}$ \\
$[\![ 30,4,6 ]\!],  [\![ 42,4,8 ]\!]$ & $\overline{\text{CNOT}^{\otimes2}}$ \\
$[\![ 70,6,10 ]\!]$ & $\overline{\text{CNOT}^{\otimes3}}$ \\
$[\![ 70,8,10 ]\!]$ & $\overline{\text{CNOT}^{\otimes4}}$ \\
$[\![ 90,10,9 ]\!]$ & $\overline{\text{CNOT}^{\otimes5}}$ \\
\bottomrule
\end{tabular}
\end{table}

\section{Girth of GB codes}\label{Ap:S5}

In graph theory, the \textit{girth} of a graph provides a fundamental measure of its cyclic structure. 
Let $G = (V, E)$ be an undirected graph, where $V$ represents the set of vertices and $E$ the set of edges. 
Then, the \textit{girth} of $G$, denoted as $g(G)$, is the length of the shortest cycle in the graph. The girth is a critical parameter because it directly impacts the performance of message-passing decoders, such as belief propagation (BP) \cite{demarti2024decoding}. 
Since many qLDPC codes are decoded using a post-processing algorithm followed by BP, making the former more robust makes decoders to be more accurate and faster.

The importance of girth arises from the role of cycles in the Tanner graph during iterative decoding. In the absence of cycles, message-passing algorithms can achieve optimal performance, as messages propagated along the graph remain uncorrelated and independent. 
However, the presence of cycles introduces correlation effects, where messages can traverse a cycle and return to their origin, creating feedback loops that distort the reliability of the decoding process. 
Shorter cycles exacerbate this problem, as they allow correlations to build up more quickly, degrading the decoder's performance \cite{degen}. 
For generalized bicycle codes, the girth is closely tied to the algebraic properties of the code (the corresponding polynomials), and careful design can yield Tanner graphs with desirable cyclic properties. 

The underlying girth conditions on GB codes can be proved using the discussion of \cite{fossorier2004quasicyclic}. 
Here we tailor such discussions to express the conditions in terms of the generator polynomials $a(x)$ and $b(x)$ of GB codes.
In particular, we characterize all codes with girths $4$, $6$, and $8$ which is the maximum possible girth of a GB code.  

For the rest of this section, we label qubits in a GB codes corresponding to polynomials $a(x)$ and $b(x)$ with $q_i$ and $q_i'$, respectively.

\begin{theorem}\label{T:girth4}
Let $ a(x)$ and $ b(x) \in \mathbb{F}_{2} [x] / \langle x^{n} - 1 \rangle$ be two polynomials that define a GB code. The Tanner graph associated with this code has girth $g(G) = 4$ if and only if there exist indices $i, j\neq j' \in \{0,1, \dots, n-1\}$ such that at least one of the following conditions holds: 
\begin{itemize}
    \item $a_j = a_{j+i} = 1$ and $a_{j'} = a_{j'+i} = 1$,
    \item $b_j = b_{j+i} = 1$ and $b_{j'} = b_{j'+i} = 1$,
    \item $a_j = b_{j+i} = 1$ and $a_{j'} = b_{j'+i} = 1$,
\end{itemize}
where all indices are taken modulo $n$.
\end{theorem}
\begin{proof}
\begin{figure}[ht]
\centering
\begin{tikzpicture}
  [scale=.5,auto=left,every node/.style={circle,fill=blue!20}]
  \node (n4) at (4,8)  {$x_0$};
  \node (n5) at (8,8)  {$p_{j+i}$};
  \node (n2) at (8,6)  {$x_{i}$ };
  \node (n3) at (4,6)  {$p_{j'+i}$};
  \foreach \from/\to in {n4/n5,n5/n2,n2/n3,n3/n4}
    \draw (\from) -- (\to);
\end{tikzpicture}
\captionsetup{width=\linewidth}
\caption{An example of a girth four. Here $x_0$ and $x_i$ are check qubits, and the rest are data qubits.}
\label{F:girth4}
\end{figure}

Each cycle of length 4 is equivalent to a graph of shown in Figure \ref{F:girth4}. 
So we consider this graph as our base case and, WLOG, we assume that the 4 cycle is happening at check qubits $0$ and $i$ for some $1 \le i \le n-1$. Such cycle is in correspondence to $(a(x),b(x))$ and $(x^ia(x),x^ib(x))$. Such 4 cycle consists of two physical qubits $p_{j+i}$ and $p_{j'+i}$. such physical qubits can be the label of non-zero coefficients of $a(x)$, $b(x)$, or one coefficient in $a(x)$ and one in $b(x)$. 
This gives precisely the three cases above.   
$\hfill \square$ \end{proof}

The conditions for a GB code to have girth six is summarized below.

\begin{theorem}\label{T:girth6}
Let $ a(x), b(x) \in \mathbb{F}_{2} [x] / \langle x^{n} - 1 \rangle$ be two polynomials that define a GB code. 
The Tanner graph associated with this code has girth $g(G) = 6$ if and only if there exist indices $i\neq j,\ell , \ell', \ell'' \in \{0,1, \dots, n-1\}$ such that at least one of the following conditions holds: 
\begin{itemize}
    \item $a_{\ell} = a_{\ell+i} = 1$, $a_{\ell'} = a_{\ell'+j} = 1$ and $a_{\ell''} = a_{\ell''+i+j} = 1$
    \item $b_{\ell} = b_{\ell+i} = 1$, $b_{\ell'} = b_{\ell'+j} = 1$ and $b_{\ell''} = b_{\ell''+i+j} = 1$
    \item $a_{\ell} = a_{\ell+i} = 1$, $a_{\ell'} = a_{\ell'+j} = 1$ and $b_{\ell''} = b_{\ell''+i+j} = 1$. 
    \item $b_{\ell} = b_{\ell+i} = 1$, $b_{\ell'} = b_{\ell'+j} = 1$ and $a_{\ell''} = a_{\ell''+i+j} = 1$. 
\end{itemize}
where all indices are taken modulo $n$. 
\end{theorem}
\begin{proof}
\begin{figure}[ht]
\centering\begin{tikzpicture}
  [scale=.5,auto=left,every node/.style={circle,fill=blue!20}]
  \node (n6) at (2,7) {$v''$};
  \node (n4) at (4,8)  {$x_{0}$};
  \node (n5) at (8,8)  {$v$};
  \node (n1) at (10,7) {$x_{i}$};
  \node (n2) at (8,6)  {$v'$};
  \node (n3) at (4,6)  {$x_{i+j}$};
  \foreach \from/\to in {n6/n4,n4/n5,n5/n1,n1/n2,n6/n3,n2/n3}
    \draw (\from) -- (\to);
\end{tikzpicture}
\captionsetup{width=\linewidth}
\caption{An example of a girth six. Here $x_0$, $x_i$, and $x_{i+j}$ are check qubits, and the rest are data qubits.}
\label{F:girth6}
\end{figure}
Again we consider the graph shown in Figure \ref{F:girth6} as our base case, and base our reasoning using that.  First, note that the case $i=j$ implies that the girth is four by Theorem \ref{T:girth4}. 
So we assume that $i \ne j$.

The connection between $v$, $x_0$, and $x_i$ implies that there exists $\ell \in \{0,1, \dots, n-1\}$ such that $a_{\ell}=a_{\ell+i}=1$ (or $b_{\ell}=b_{\ell+i}=1$). 
In this case, $v=a_{\ell+i}$ (respectively $v=b_{\ell+i}$). 

The connection between $v'$, $x_i$, and $x_{i+j}$ implies that there exists $\ell' \in \{0,1, \dots, n-1\}$ such that $a_{\ell'}=a_{\ell'+j}=1$ (or $b_{\ell'}=b_{\ell'+j}=1$). In this case, $v'=a_{\ell'+i+j}$ (respectively $v'=b_{\ell'+i+j}$). 

Finally, the connection between $v''$, $x_{i+j}$, and $x_{0}$ implies that there exists $\ell'' \in \{0,1, \dots, n-1\}$ such that $a_{\ell''}=a_{\ell''+i+j}=1$ (or $b_{\ell''}=b_{\ell''+i+j}=1$). In this case, $v''=a_{\ell''+i+j}$ (respectively $v'=b_{\ell''+i+j}$).

A combination of the above possibilities gives one of the four cases given in the theorem. 
$\hfill \square$ \end{proof}

Next we show that the girth of a GB code cannot exceed eight. The proof follows in two steps. First we have the following lemma. 

\begin{lemma} \label{L:girth8}
Let $a(x) = 1 + x^{i}$ and $b(x) = x^{j} + x^{j+k}$ be two polynomials with weight two in $\mathbb{F}_{2} [x] / \langle x^{n} - 1 \rangle$. Then, the girth of the Tanner graph associated with the GB code defined by $a(x)$ and $b(x)$ satisfies $g(G) \leq  8.$
\end{lemma}
\begin{proof}
Note that each check node \( x_t \) is connected to four physical qubits, namely:
\begin{equation*}
    q_t, \quad q_{t+i}, \quad q'_{t+j}, \quad q'_{t+j+k}.
\end{equation*}
More specifically, we have
\begin{itemize}
    \item A qubit node $q_{t}$ (corresponding to $t$-th coefficient of $a(x)$) is connected to check nodes $x_{t}$ and $x_{t-i}$.
    \item A qubit node $q^{'}_{t}$ (corresponding to $t$-th coefficient of $b(x)$) is connected to check nodes  $x_{t-j}$ and $x_{t-j-k}$.
\end{itemize}
Now, we construct an explicit cycle of length 8 as: 
\begin{center}
\begin{tikzpicture}
  [scale=.8,auto=left,every node/.style={circle,fill=blue!20}]
  
  \node (n0) at (0,0)    {$q_{0}$};
  \node (n1) at (2,0)    {$x_{0}$};
  \node (n2) at (4,0)    {$q_{j}'$};
  \node (n3) at (6,0)    {$x_{-k}$};
  \node (n4) at (6,-2)   {$q_{-k}$};
  \node (n5) at (4,-2)   {$x_{-k-i}$};
  \node (n6) at (2,-2)   {$q^{'}_{j-i}$};
  \node (n7) at (0,-2)   {$x_{-i}$};

  \foreach \from/\to in {n0/n1,n1/n2,n2/n3,n3/n4,n4/n5,n5/n6,n6/n7,n7/n0}
    \draw (\from) -- (\to);
\end{tikzpicture}
\end{center}
Thus, $g(G) \leq 8.$
$\hfill \square$ \end{proof}

The following corollary is a straightforward consequence of the previous Lemma. 

\begin{corollary} \label{C:girth8b}
Let $a(x)$ and $ b(x) \in \mathbb{F}_{2}[x] / \langle x^n - 1 \rangle$ with weights at least two. Then, the girth of the Tanner graph associated with such GB code satisfies
 $g(G) \leq 8.$
\end{corollary}

\begin{proof}
 The proof follows from the previous lemma and the fact that Tanner graph of such GB code has a subgraph which is the Tanner graph of the GB code of Lemma \ref{L:girth8}.   
$\hfill \square$ \end{proof} 

Using the above results, one can show that the girth eight case is achievable by the GB codes discussed in the main text. 

\section{Circuit-level noise simulation}\label{Ap:S6}
The noise model considered in this work is the \emph{circuit-level depolarizing noise model}, which is considered one of the most realistic models in the literature, with physical error rate $p$. 
In this model, errors are applied independently after every elementary operation in the syndrome extraction circuit, and the syndrome extraction circuit is designed using the RL pattern to avoid the hook errors. 
The noise processes are defined as follows:

\begin{itemize}
    \item \textbf{Measurement errors:} Every measurement outcome is flipped with probability $p$.
    
    \item \textbf{Preparation errors:} Every qubit preparation results in the orthogonal state with probability $p$.
    
    \item \textbf{Gate and idling errors:} Every single-qubit and two-qubit unitary gate is followed by a depolarizing noise channel of strength $p$. 
    In addition, during each time step of the circuit, qubits that are not involved in any gate or measurement are assumed to undergo an \emph{idle} operation. These idle locations are also followed by a single-qubit depolarizing channel of strength $p$.
\end{itemize}

More precisely, for each discrete time slice of the circuit, we assume that every qubit participates in exactly one operation: either a gate, a measurement, a preparation, or an idle operation. 
Thus, idling errors are applied uniformly across all qubits not actively involved in other operations.

The single-qubit and two-qubit depolarizing channels of strength $p$ are defined, respectively, as
\begin{align*}
\mathcal{E}_p^{(1)}(\rho^{(1)}) 
&= (1 - p)\rho^{(1)} + \frac{p}{3} \sum_{P \in \{X,Y,Z\}} P \rho^{(1)} P, 
\end{align*}
and 
\begin{align*}
\mathcal{E}_p^{(2)}(\rho^{(2)}) 
&= (1 - p)\rho^{(2)} + \frac{p}{15} \times\\&\sum_{\substack{P_1,P_2 \in \{I,X,Y,Z\} \\ (P_1,P_2) \neq (I,I)}} (P_1 \otimes P_2)\rho^{(2)}(P_1 \otimes P_2),
\end{align*}
where $\rho^{(1)}$ and $\rho^{(2)}$ denote arbitrary single-qubit and two-qubit density operators, respectively.
All error processes are assumed to be independent and identically distributed (iid), and the same physical error rate $p$ is used for all types of operations, yielding a symmetric and hardware-agnostic circuit-level noise model.

In our threshold estimation, we applied $d$ rounds of syndrome extraction circuit 
and applied a Monte Carlo simulations. In particular,
for each code instance and each physical error rate, the logical error rate was estimated using a maximum of $10^7$ shots (respectively $10^6$ shots) for MWPM (respectively BP-OSD) or $500$ logical errors, whichever is reached first.

We also conducted a second circuit-level simulation under the same set of assumptions except after removing the idling errors under MWPM and executing $10^{10}$ number of shots. 
This led to a boost in the performance and threshold (leading to a threshold of approximately $1\%$) for both families. 
In particular, the performance of the even distance family after removing the idling error is provided in Figure \ref{fig:logical_error_rates3}.

\begin{figure}[htbp]
\includegraphics[width=0.47\textwidth]{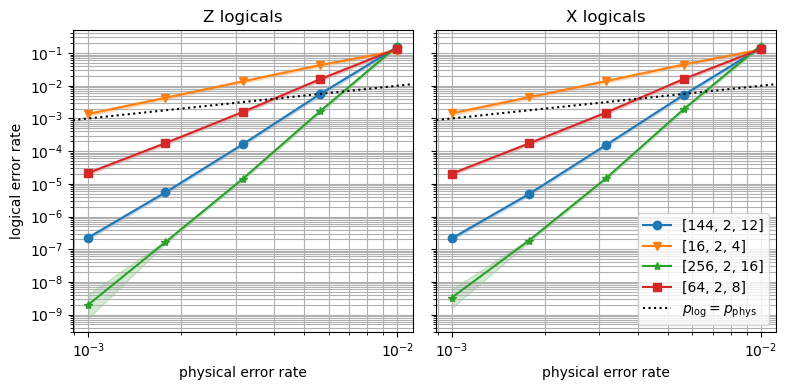}
    \caption{
        Logical versus physical error rate for the $[\![d^2,2,d]\!]_2$ GB code family with $d \in \{4,8,12,16\}$, using the MWPM decoder under circuit-level depolarizing noise after removing the idling error.}
    \label{fig:logical_error_rates3}
\end{figure}

The close similarity in threshold performance between GB codes and surface codes highlights the strong potential of GB codes as viable alternatives for fault-tolerant quantum computing. 
This result is particularly notable given the favorable properties of GB codes, including higher encoding rate, low-weight parity checks, existence of redundant checks, and a highly regular structure. 

\end{document}